\documentclass[11pt]{article}
\usepackage[margin=1in]{geometry}

\usepackage{caption}
\usepackage{subcaption}
\usepackage{amsmath, amssymb} 
\usepackage{tikz}
\usetikzlibrary{arrows.meta,positioning,calc,backgrounds}
\usepackage{mathtools}
\usepackage{amsthm}
\usepackage{enumitem}

\usepackage[authoryear,round]{natbib}
\usetikzlibrary{calc,positioning}

\usepackage{algorithm}
\usepackage{algorithmic}
\usepackage{tabularx}

\newtheorem{Theorem}{Theorem}
\newtheorem{Proposition}{Proposition}
\newtheorem{Lemma}{Lemma}

\newtheorem{Definition}{Definition}
\newtheorem{Corollary}{Corollary}
\newtheorem{Assumption}{Assumption}
\newtheorem{Example-set}{Example}

\newtheorem{Remark}{Remark}

\usepackage[dvipsnames]{xcolor}        % colors

\usepackage[colorlinks=true, linkcolor={rgb,255:red,0;green,20;blue,115}, citecolor={rgb,255:red,0;green,20;blue,115}]{hyperref}

\usepackage{breakurl}

\usepackage{cleveref}

\usepackage{enumitem}
\usepackage{bbding}
\usepackage{bm}

\crefformat{subfigure}{#2Fig.~#1#3}
\Crefformat{subfigure}{#2Fig.~#1#3}

\crefname{theorem}{Theorem}{Theorems}
\crefname{lemma}{Lemma}{Lemmas}
\crefname{property}{Property}{Properties}
\crefname{assumption}{Assumption}{Assumptions}
\crefname{proposition}{Proposition}{Propositions}
\crefname{property}{Property}{Properties}
\crefname{corollary}{Corollary}{Corollaries}
\crefname{figure}{Fig.}{Figs.}
\crefname{section}{Section}{Sections}
\crefname{definition}{Definition}{Definitions}
\crefname{table}{Table}{Tables}

\crefname{algorithm}{Algorithm}{Algorithms}
\crefname{example}{Example}{Examples}
\crefname{remark}{Remark}{Remarks}
\crefname{appendix}{Appendix}{Appendices}
\crefname{section}{Section}{Sections}

\usepackage{float}
\usepackage{placeins}

\usepackage{titletoc}

\usepackage[textsize=tiny]{todonotes}

\usepackage{adjustbox}

\usepackage{centernot}

\newcommand{\CI}{\mathrel{\perp\mspace{-10mu}\perp}}
\newcommand{\nCI}{\centernot{\CI}}

\newdimen\arrowsize

\pgfarrowsdeclare{arcsq}{arcsq}
{
	\arrowsize=0.5pt
	\advance\arrowsize by .5\pgflinewidth
	\pgfarrowsleftextend{-10\arrowsize-.10\pgflinewidth}
	\pgfarrowsrightextend{.10\pgflinewidth}
}
{
	\arrowsize=0.5pt
	\advance\arrowsize by .5\pgflinewidth
	\pgfsetdash{}{0pt} % do not dash
	\pgfsetroundjoin   % fix join
	\pgfsetroundcap    % fix cap
	\pgfpathmoveto{\pgfpoint{0\arrowsize}{0\arrowsize}}
	\pgfpatharc{-90}{-130}{7\arrowsize}
	\pgfusepathqstroke
	\pgfpathmoveto{\pgfpointorigin}
	\pgfpatharc{90}{130}{7\arrowsize}
	\pgfusepathqstroke
}
\pgfarrowsdeclare{arcsqblue}{arcsqblue}
{
	\arrowsize=0.5pt % Increase arrow thickness
	\advance\arrowsize by .5\pgflinewidth
	\pgfarrowsleftextend{-12\arrowsize-.10\pgflinewidth}
	\pgfarrowsrightextend{.10\pgflinewidth}
}
{
	\arrowsize=0.8pt % Increase arrow thickness
	\advance\arrowsize by .5\pgflinewidth
	\pgfsetdash{}{0pt} % do not dash
	\pgfsetroundjoin   % fix join
	\pgfsetroundcap    % fix cap
	\pgfsetcolor{blue} % Set the color to blue
	\pgfsetlinewidth{1.5pt} % Set line width
	\pgfpathmoveto{\pgfpoint{0\arrowsize}{0\arrowsize}}
	\pgfpatharc{-90}{-130}{6\arrowsize}
	\pgfusepathqstroke
	\pgfpathmoveto{\pgfpointorigin}
	\pgfpatharc{90}{130}{6\arrowsize}
	\pgfusepathqstroke
}

\tikzset{
    every path/.style={line width=0.75pt},
    every node/.style={font=\normalsize,inner sep=0pt,
        text=black,
        text opacity=1},
    selection/.style={
        draw,
        rectangle,
        fill=black!50,
        minimum width=5.5mm,
        minimum height=5.5mm,text=white,
    },
    latent/.style={
        draw,
        circle,
        dash pattern={on 3.0pt off 2.5pt},
        fill=black!50,
        line width=1.0pt, minimum size=5.5mm,text=white,
    },
    obs/.style={
        circle, draw=Gray, fill=Gray!10, line width=0.8pt, minimum size=5.5mm
    },
    geneobs/.style={
        rectangle,
        rounded corners=2pt,
        draw=Gray,
        fill=Gray!10,
        line width=0.8pt,
        minimum width=7.5mm,
        minimum height=5.5mm
    },
   target/.style={
        circle, 
        draw=OrangeRed, 
        fill=OrangeRed!10, 
        line width=1.5pt, 
        minimum size=5.5mm, 
        inner sep=0pt
   },
   genetarget/.style={
        rectangle,
        rounded corners=2pt,
        draw=OrangeRed,
        fill=OrangeRed!10,
        line width=1.5pt,
        minimum width=7.5mm,
        minimum height=5.5mm,
        inner sep=0pt
    },
    mb/.style={
        circle, 
        draw=RoyalBlue, 
        fill=RoyalBlue!10, 
        line width=1pt, 
        minimum size=5.5mm, 
        inner sep=0pt
    }
    }
\usetikzlibrary{arrows.meta}

\makeatletter
\pgfarrowsdeclare{circle}{circle}
{
    \pgfarrowsleftextend{0pt}
    \pgfarrowsrightextend{6\pgflinewidth}
}
{
    \pgfpathcircle{\pgfqpoint{1.5pt}{0pt}}{1.5pt} % Adjust this value for size (e.g., 3pt)
    \pgfusepathqstroke % Outline the circle without filling
}
\makeatother

\usepackage{xspace}
\newcommand{\wrt}{{w.r.t.\@\xspace}}         % wrt.,
\newcommand{\eg}{e.g.\@\xspace}           % e.g.
\newcommand{\ie}{i.e.\@\xspace}           % i.e.
\usepackage{threeparttable}
\usepackage{booktabs}    % Standard package for three-line tables
\usepackage{multirow}    % Handles merged rows
\usepackage{xcolor}      % Color support
\usepackage{colortbl}    % Table background color support
\usepackage{pifont}      % Provides check and cross marks
\usepackage{fontawesome5}% Provides GitHub and web icons
\usepackage{arydshln} % Dashed lines
\usepackage{array}
\newcolumntype{C}[1]{>{\centering\arraybackslash}m{#1}}

\definecolor{lightblue}{RGB}{220,240,244}

\newlength{\localfigbodyheight}
\newlength{\seqfigbodyheight}
\date{}

\begin{document}

\title{
Testing the Validity of Instrumental Variable Sets in Causal Additive Models with Non-Constant Effects 
}

\author{Xichen Guo\textsuperscript{1}, Feng Xie\textsuperscript{*1}, Bingbing Tang\textsuperscript{1}, Yan Zeng\textsuperscript{1}, Zhang Hao\textsuperscript{2}, \\Zhi Geng\textsuperscript{1,3}, Ruichu Cai\textsuperscript{4,5}, and Kun Zhang\textsuperscript{6,7}%
}

\maketitle

\begingroup
\makeatletter
\insert\footins{%
    \reset@font\footnotesize
    \interlinepenalty\interfootnotelinepenalty
    \splittopskip\footnotesep
    \splitmaxdepth\dp\strutbox
    \floatingpenalty\@MM
    \hsize\columnwidth
    \@parboxrestore
    \noindent
    \textsuperscript{1}Department of Applied Statistics, Beijing Technology and Business University, Beijing, China.
    \textsuperscript{2}Shenzhen Institutes of Advanced Technology, Chinese Academy of Sciences, Shenzhen, China.
    \textsuperscript{3}School of Mathematical Sciences, Peking University, Beijing, China.
    \textsuperscript{4}School of Computer Science, Guangdong University 
    of Technology, Guangzhou, China. 
    \textsuperscript{5}Peng Cheng Laboratory, Shenzhen, China. 
    \textsuperscript{6}Department of Philosophy, Carnegie Mellon University, Pittsburgh, USA.
    \textsuperscript{7}Mohamed bin Zayed University of Artificial Intelligence, Abu Dhabi, United Arab Emirates.
    *Corresponding author: Feng Xie (e-mails: fengxie@btbu.edu.cn).\par
}
\makeatother
\endgroup

\begin{abstract}
Instrumental variable (IV) methods are powerful for causal effect estimation with unmeasured confounding, but in practice researchers often face a set of candidate IVs whose validity is difficult to determine from observational data.
This paper studies the problem of testing the validity of IV sets under Causal Additive Models with Non-Constant Effects (CAM-NCE). 
To address this problem, we propose a testable condition, termed the {Cross Auxiliary-based independence Test} (CAT) condition, for assessing IV set validity from observational data. 
We show that, under the completeness condition,  if the CAT condition is violated, the corresponding set cannot be a valid IV set. 
Furthermore, under a cross distributional non-degeneracy condition, we establish that the CAT condition becomes both necessary and sufficient for characterizing valid IV sets under CAM-NCE. 
We then extend the CAT condition to settings with covariates and develop a practical finite-sample algorithm for testing the validity of candidate IV sets. 
Extensive experiments on synthetic data and three real-world datasets demonstrate the effectiveness and practical utility of the proposed method.
\end{abstract}


\section{Introduction}\label{Sec-introduction}
%----------------------%
Causal effect estimation from observational data is a fundamental problem in modern machine learning and probabilistic inference~\citep{peters2011causal,Cheng2023causal}, with broad applications in recommendation systems~\citep{cao2023mitigating,zhang2023personalized,liao2024instrumental}, online systems and sequential decision-making~\citep{shi2021temporal,lin2023towards,zheng2024causal}, and causality-aware clustering~\citep{li2026mechanisms,kim2024hierarchical,kim2026causal,cao2025fwcec}. 
% \IEEEPARstart{C}{ausal} effect estimation is increasingly used in modern learning systems to support reliable decision-making from observational data~\cite{Cheng2023causal,kun2024causal}, with applications in recommendation systems~\cite{ding2024causal,zhu2024mitigating,zheng2024causal}, online systems and sequential decision-making~\cite{zeng2025causal,liao2024instrumental,Fujiicausal}, and causality-aware clustering~\cite{kim2024hierarchical,kim2026causal,cao2025fwcec}. 
Instrumental variable (IV) methods provide a powerful approach to causal effect estimation by leveraging exogenous variables to mitigate unmeasured confounding. Informally, a valid IV must satisfy three key conditions: 
($\mathcal{C}$1) it is associated with the treatment (relevance); 
($\mathcal{C}$2) it affects the outcome only through the treatment (exclusion restriction); and 
($\mathcal{C}$3) it is independent of unmeasured confounders (exogeneity) 
(see Figure~\ref{fig-one-IV-assumptions-example} and Section~\ref{Sec-Notations and Definitions} for formal definitions). 
For instance, in Figure~\ref{fig-one-IV-assumptions-example}, \emph{Colonialist Mortality} serves as a valid IV for estimating the causal effect of \emph{Institutions} on \emph{Economic Development}. 

Because of unmeasured confounding, it is generally difficult to determine from observational data alone which variables can serve as valid IVs. 
In practice, valid IVs may sometimes be justified \emph{a priori} by domain expertise or historical knowledge. 
However, such information is often unavailable in many empirical settings, leaving researchers with a set of candidate IVs whose validity is uncertain. 
Therefore, developing data-driven methods to test IV validity remains a critical yet challenging task~\citep{cheng2024data,wu2025instrumental}.

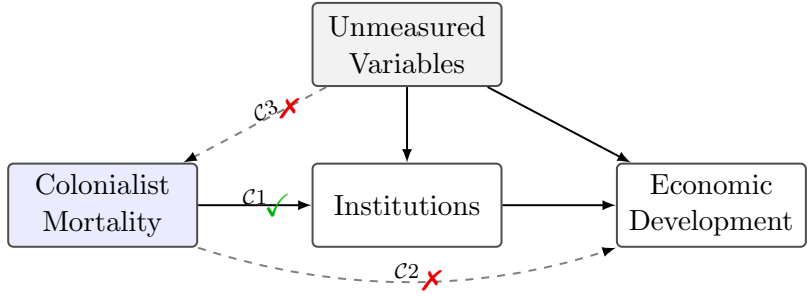
\begin{figure}[!t]
\centering
\begin{tikzpicture}[
    node distance=1.3cm and 1.5cm,
    line width=0.8pt,
    >=latex,
    font=\footnotesize,
    ivnode/.style={rectangle, rounded corners=2pt, draw=black!70, fill=blue!8,
                   text width=2.5cm, align=center, minimum height=1.1cm},
    uinode/.style={rectangle, rounded corners=2pt, draw=black!70, fill=gray!10,
                   text width=2.5cm, align=center, minimum height=1.1cm},
    crossmark/.style={text=red!90!black, font=\scriptsize},
    dashedarrow/.style={dashed, draw=black!50, ->},
    solidarrow/.style={draw=black, ->}
]

% ----------- -----------
\node[ivnode] (G) {Colonialist\\Mortality};
\node[uinode, right=of G, fill=white] (X) {Institutions};
\node[uinode, right=of X, fill=white] (Y) {Economic\\Development};
\node[uinode, above=1.0cm of X] (U) {Unmeasured\\Variables};

\draw[solidarrow] (G) -- node[above,sloped,pos=0.5] {\scriptsize $\mathcal{C}1$} (X);
\draw[solidarrow] (U) -- (X);
\draw[solidarrow] (U) -- (Y);
\draw[solidarrow] (X) -- (Y);

\draw[dashedarrow] (U) -- node[above,sloped,pos=0.4] {\scriptsize $\mathcal{C}3$} (G);

\draw[dashedarrow, bend right=18, looseness=0.9]
    (G.south east) to
    node[sloped, pos=0.5, above, inner sep=1pt, align=center]
    {\scriptsize $\mathcal{C}2$}
    (Y.south west);

\node at ($(G)!0.5!(X)+(0.3,0.00)$) {\textcolor{green!70!black}{\large{\checkmark}}}; 
\node[crossmark] at ($(U)!0.4!(G)+(0.05,0.0)$) {\XSolidBrush};
\node[crossmark] at ($(G)!0.54!(Y)+(0.00,-1.0)$) {\XSolidBrush};
\end{tikzpicture}
% \vspace{-2mm}
\caption{Graphical illustration of a valid IV model.
Solid arrows denote causal relationships, while dashed arrows denote prohibited relationships for IV validity.
The variable \emph{Colonialist Mortality} serves as a valid IV for the causal effect of
\emph{Institutions} on \emph{Economic Development}. Unmeasured
Variables (\eg, \emph{Cultural
Difference}) act as unmeasured confounders \citep{acemoglu2001colonial}. More details are provided in Section~\ref{Section-real-world}.}
\label{fig-one-IV-assumptions-example}
\end{figure}

%------------------Existing methods for discrete data-------------%

Testing the validity of IVs has attracted considerable attention in recent years. 
For discrete treatment settings, several approaches address this problem by imposing testable constraints on the joint distribution, including Pearl's seminal instrumental inequality~\citep{pearl1995testability}, the generalized instrumental inequality~\citep{kedagni2020generalized}, and their extensions~\citep{manski2003partial,palmer2011nonparametric,kitagawa2015test,wang2017falsification,huber2015testing,mourifie2017testing,farbmacher2022instrument}. 
Although these methods have been successfully applied in various domains, they typically rely on discrete treatment variables, which limits their applicability in many practical settings—such as studies involving continuous exposures (e.g., vitamin levels~\citep{skaaby2013vitamin}).

%------------------Existing methods for continous data-------------%

Another line of research studies IV validity in the Causal Additive Models with Constant Effects (CAM-CE). Existing approaches in this line can be broadly divided into two categories.
\begin{itemize}
    \item \textbf{IV Set.} 
    Representative methods include the \emph{Proportion of Invalid IVs} approaches, such as the \emph{Majority Rule} constraint, which assumes that more than $50\%$ of the candidate IVs are valid~\citep{han2008detecting,kang2016instrumental,bowden2016consistent,windmeijer2019use,hartford2021valid}, and the \emph{Plurality Rule} constraint, which requires the number of valid IVs exceeds that of any group of invalid IVs sharing the same ratio-estimator limit~\citep{hartwig2017robust,guo2018confidence,windmeijer2021confidence,lin2024instrumental}. 
    Another representative condition is the \emph{InSIDE constraint}, which assumes that the pleiotropic effects of IVs on the outcome are uncorrelated with their effects on the treatment~\citep{bowden2015mendelian,kolesar2015identification,sanderson2022mendelian}. 
    In addition, the \emph{Two Valid IVs} constraint requires the existence of at least two valid instruments and further relies on either the \emph{Rank-Faithfulness} assumption~\citep{silva2017learning,cheng2023discovering} or an \emph{algebraic condition} discussed in~\cite{guo2025data}.

    \item \textbf{Single IV.} 
    Typical approaches include the IV-GIN method, which leverages the Generalized Independent Noise (GIN) condition~\citep{xie2020generalized} within linear non-Gaussian causal models~\citep{xie2022testability}, and the IV-PIM method, which is based on the Principle of Independent Mechanisms (PIM)~\citep{janzing2012information}  and evaluates instrument validity by decomposing the spectral measure of the covariance matrix of the covariates~\citep{burauel2023evaluating}.
\end{itemize}
However, these approaches primarily focus on constant causal effects, and their theoretical guarantees do not directly extend to settings with non-constant causal effects.

Causal additive models with non-constant effects (CAM-NCE) provide a more flexible framework for capturing heterogeneous treatment--outcome relationships. 
Chu et al.~\citep{chu2001semi} were among the first to study IV validity in this setting by introducing the concept of \emph{semi-instruments}, showing that semi-instrument validity is testable under additive models. 
However, their testability condition relies on the prior assumption that the exogeneity condition ($\mathcal{C}$3) holds. 
More recently, \cite{guo2024testability} proposed the Auxiliary-based Independence Test (AIT) for a single IV and showed that, under certain conditions, a valid IV satisfies the AIT condition. 
Nevertheless, the AIT condition encounters difficulties in verifying the exclusion restriction ($\mathcal{C}$2) (see the illustrative example in Section~\ref{Subsection-CAT} and Proposition~5 in ~\citep{guo2024testability}). 
Moreover, the exclusion restriction is often difficult to satisfy in practice. 
For example, in Mendelian randomization studies, genetic variants may have pleiotropic effects on the outcome, thereby violating ($\mathcal{C}$2)~\citep{burgess2017review}.

\begin{table}[h]
\centering
\caption{Summary of research coverage on IV validity testing under continuous-variable settings.}
\label{tab:IV-validity-summary}

\renewcommand{\arraystretch}{1.25}
\setlength{\tabcolsep}{4pt}
\setlength{\arrayrulewidth}{0.4pt}

\small
\begin{tabular}{%
    >{\raggedright\arraybackslash}m{0.14\linewidth} |
    >{\centering\arraybackslash}m{0.52\linewidth} |
    >{\centering\arraybackslash}m{0.27\linewidth}
}
\toprule

\textbf{Category}
&
\textbf{Constant Effect}
&
\textbf{Non-Constant Effect}
\\
\hline

Single IV
&
\textcolor{green!70!black}{\large\checkmark}\;
\cite{xie2022testability}; and \cite{burauel2023evaluating}
&
\textcolor{green!70!black}{\large\checkmark}\;
\cite{chu2001semi}; and \cite{guo2024testability}
\\[2pt]
\hline

IV Set
&
\textcolor{green!70!black}{\large\checkmark}\;
\cite{han2008detecting,kang2016instrumental,bowden2016consistent,windmeijer2019use,hartford2021valid, hartwig2017robust,guo2018confidence,windmeijer2021confidence,lin2024instrumental, bowden2015mendelian,kolesar2015identification,sanderson2022mendelian, silva2017learning,cheng2023discovering}; and \cite{guo2025data}
&
\textcolor{red!90!black}{\large\XSolidBrush}
\\

\bottomrule
\end{tabular}

\vspace{2mm}

\raggedright
\footnotesize
\textit{Note:}
The check mark
\textcolor{green!70!black}{\checkmark}
indicates that studies have addressed the corresponding setting, whereas the cross
\textcolor{red!90!black}{\XSolidBrush}
denotes that research in this setting remains unexplored.
\end{table}

% \begin{table}[h]
% \centering
% \caption{Summary of research coverage on IV validity testing under continuous-variable settings.}
% \label{tab:IV-validity-summary}
% \renewcommand{\arraystretch}{1.25}
% \setlength{\tabcolsep}{4pt}
% \setlength{\arrayrulewidth}{0.4pt}

% \small
% \begin{tabular}{%
%     >{\raggedright\arraybackslash}m{0.14\linewidth} |
%     >{\centering\arraybackslash}m{0.52\linewidth} |
%     >{\centering\arraybackslash}m{0.27\linewidth}
% }
% \toprule
% \textbf{Category}
% &
% \textbf{Constant Effect}
% &
% \textbf{Non-Constant Effect}
% \\
% \midrule

% Single IV
% &
% \textcolor{green!70!black}{\large\checkmark}\;
% \cite{xie2022testability}; and \cite{burauel2023evaluating}
% &
% \textcolor{green!70!black}{\large\checkmark}\;
% \cite{chu2001semi}; and \cite{guo2024testability}
% \\[2pt]

% IV Set
% &
% \textcolor{green!70!black}{\large\checkmark}\;
% \cite{han2008detecting,kang2016instrumental,bowden2016consistent,windmeijer2019use,hartford2021valid,hartwig2017robust,
% guo2018confidence,windmeijer2021confidence,lin2024instrumental,bowden2015mendelian,kolesar2015identification,sanderson2022mendelian,silva2017learning,cheng2023discovering}; and \cite{guo2025data}
% &\textcolor{red!90!black}{\large\XSolidBrush}
% \\

% \bottomrule
% \end{tabular}

% \vspace{2mm}

% \raggedright
% \footnotesize
% \textit{Note:}
% The check mark
% \textcolor{green!70!black}{\checkmark}
% indicates that studies have addressed the corresponding setting,
% whereas the cross
% \textcolor{red!90!black}{\XSolidBrush}
% denotes that research in this setting remains unexplored.
% \end{table}

\cref{tab:IV-validity-summary} summarizes existing studies on IV validity testing under continuous-variable settings, categorized by whether the analysis involves a single IV or a set of IVs, and whether the underlying causal effect is constant or non-constant. 
As shown in the table, prior research has primarily focused on either single IV settings or IV sets under constant-effect assumptions. 
However, testing the validity of IV sets—particularly when the treatment–outcome relationship exhibits non-constant causal effects—remains underexplored. 
To bridge this gap, this paper investigates the problem of testing the validity of IV sets within the CAM-NCE framework, where candidate instruments may violate either the exclusion restriction or the exogeneity condition.
Specifically, our main contributions are as follows:
\begin{enumerate}
    \item[1.] We introduce a testable condition, termed the \textbf{C}ross \textbf{A}uxiliary-based Independence \textbf{T}est (CAT) condition, for assessing the validity of IV sets within CAM-NCE.
    
    \item[2.] We establish that, under the completeness condition (Assumption~\ref{Ass-completeness}) and the cross distributional non-degeneracy condition (Assumption~\ref{Ass-higher-order-condition}), the CAT condition is necessary and sufficient for detecting all invalid IV sets under CAM-NCE. 
    
    \item[3.] We develop a practical testing procedure for the CAT condition that accounts for covariates and finite-sample considerations.
    
    \item[4.] We empirically validate the proposed method through extensive experiments on both synthetic and real-world datasets, demonstrating its effectiveness in detecting invalid IV sets.
\end{enumerate}

\section{Preliminaries and Problem Definition}\label{Sec-Preliminaries}
% \vspace{-1mm}
\subsection{Notations and Definitions}\label{Sec-Notations and Definitions}

In this paper, a causal system is represented by a \emph{directed acyclic graph (DAG)} $\mathcal{G}$, in which nodes correspond to random variables and directed edges represent direct causal influences between them. For brevity, we use ``w.r.t.'' to denote ``with respect to''. Throughout this paper, the major symbols and notations are summarized in~\cref{Table-symbols}.

\begin{table}[htp]
\centering
\small
\setlength{\tabcolsep}{1pt} 

\caption{Symbols and notations used in this paper.}
\begin{tabular}{p{3.2cm}p{10cm}}
\toprule
\textbf{Symbol} & \textbf{Description}\\ 
\midrule
% $\mathcal{G}$   & A directed acyclic graph (DAG)  \\
% IV    & Instrumental Variable (or instrument)  \\
$X$    & A treatment (exposure) variable \\
$Y$    & An outcome variable  \\
$Z_i$  & A candidate (potential) IV  \\
$\mathbf{Z}$ & A candidate (potential) IV set  \\
$\mathbf{U}$ & Unmeasured confounders between $X$ and $Y$ \\
$\mathbf{W}$    & Covariates   \\
$\mathcal{X}$    & The residual of variable $X$ after regressing on covariates $\mathbf{W}$ \\
$\mathcal{Z}_i$    & The residual of variable $Z_i$ after regressing on covariates $\mathbf{W}$ \\
$|\mathbf{Z}|$  & The number of variables in set $\mathbf{Z}$\\
$\varepsilon_{*}$ & The noise term of a variable \\
$\mathbb{E}(X)$ & The expected value of random variable $X$\\
$\operatorname{Cov}(X, Y)$ & The covariance between random variables $X$ and $Y$\\
$\operatorname{Var}(X)$ & The variance of random variable $X$\\
${A} \CI {B}$  & $A$ is statistically independent of $B$ \\
${A} \nCI {B}$ & $A$ is statistically dependent on $B$ \\
$\mathbb{R}$ &  The field of real numbers\\
$\mathbb{R}^{|*|}\rightarrow\mathbb{R}$ & A mapping from $\mathbb{R}^{|*|}$ to $\mathbb{R}$ \\
$f_{bias}(X)$ & The bias between estimated causal effect of $X$ on $Y$ and ground-truth causal effect of $X$ on $Y$ \\
$\mathcal{V}_{X \to Y||Z_i}$ & The auxiliary variable of causal relationship $X \to Y$ relative to $Z_i$. We often use $\mathcal{V}_i$ as a shorthand when there is no ambiguity \\ 
$\widehat{Q}$ & The sample estimate of the corresponding population quantity $Q$ \\
\bottomrule
\end{tabular}

\label{Table-symbols}
\end{table}

\begin{figure*}[t!]
\centering
%%%%%%%%%%%%%%
%% (a) Valid IV set
%%%%%%%%%%%%%%
\begin{tikzpicture}[scale=1.05, line width=0.8pt, inner sep=0.4mm, shorten >=.1pt, shorten <=.1pt]
\tikzstyle{every node}+=[inner sep=0pt]
% ----- 节点 -----
\draw (1.1, 1.25) node(Z2) [circle, draw=black, fill=white, minimum size=0.55cm] {\footnotesize $Z_2$};
\draw (0.3, 0.0)  node(Z1) [circle, draw=black, fill=white, minimum size=0.55cm] {\footnotesize $Z_1$};
\draw (1.9, 0.0)  node(X)  [circle, draw=black, fill=white, minimum size=0.55cm] {\footnotesize $X$};
\draw (3.5, 0.0)  node(Y)  [circle, draw=black, fill=white, minimum size=0.55cm] {\footnotesize $Y$};
\draw (2.7, 1.25) node(U)  [circle, draw=black, fill=gray!15, minimum size=0.55cm] {\footnotesize $\mathbf{U}$};

\draw[-arcsq] (Z1) -- (X);
\draw[-arcsq] (Z2) -- (X);
\draw[-arcsq] (X) -- (Y);
\draw[-arcsq] (U) -- (X);
\draw[-arcsq] (U) -- (Y);

\draw (2.0, -0.8) node {\small (a) Valid IV set };
\end{tikzpicture}
\hspace{5mm}
%%%%%%%%%%%%%%
%% (b) Violates C2
%%%%%%%%%%%%%%
\begin{tikzpicture}[scale=1.05, line width=0.8pt, inner sep=0.4mm, shorten >=.1pt, shorten <=.1pt]
\tikzstyle{every node}+=[inner sep=0pt]

\draw (1.1, 1.25) node(Z2) [circle, draw=black, fill=white, minimum size=0.55cm] {\footnotesize $Z_2$};
\draw (0.3, 0.0)  node(Z1) [circle, draw=black, fill=white, minimum size=0.55cm] {\footnotesize $Z_1$};
\draw (1.9, 0.0)  node(X)  [circle, draw=black, fill=white, minimum size=0.55cm] {\footnotesize $X$};
\draw (3.5, 0.0)  node(Y)  [circle, draw=black, fill=white, minimum size=0.55cm] {\footnotesize $Y$};
\draw (2.7, 1.25) node(U)  [circle, draw=black, fill=gray!15, minimum size=0.55cm] {\footnotesize $\mathbf{U}$};

\draw[-arcsq] (Z1) -- (X);
\draw[-arcsq] (Z2) -- (X);
\draw[-arcsq] (X) -- (Y);
\draw[-arcsq] (U) -- (X);
\draw[-arcsq] (U) -- (Y);
\draw[-arcsq, red!90!black, line width=1pt] (Z1) edge[bend right=25] (Y);

\draw (2.0, -0.8) node {\small (b) Invalid IV set violating $\mathcal{C}2$ };
\end{tikzpicture}
\hspace{5mm}
%%%%%%%%%%%%%%
%% (c) Violates C3
%%%%%%%%%%%%%%
\begin{tikzpicture}[scale=1.05, line width=0.8pt, inner sep=0.4mm, shorten >=.1pt, shorten <=.1pt]
\tikzstyle{every node}+=[inner sep=0pt]

\draw (1.1, 1.25) node(Z2) [circle, draw=black, fill=white, minimum size=0.55cm] {\footnotesize $Z_2$};
\draw (0.3, 0.0)  node(Z1) [circle, draw=black, fill=white, minimum size=0.55cm] {\footnotesize $Z_1$};
\draw (1.9, 0.0)  node(X)  [circle, draw=black, fill=white, minimum size=0.55cm] {\footnotesize $X$};
\draw (3.5, 0.0)  node(Y)  [circle, draw=black, fill=white, minimum size=0.55cm] {\footnotesize $Y$};
\draw (2.7, 1.25) node(U)  [circle, draw=black, fill=gray!15, minimum size=0.55cm] {\footnotesize $\mathbf{U}$};

\draw[-arcsq] (Z1) -- (X);
\draw[-arcsq] (Z2) -- (X);
\draw[-arcsq] (X) -- (Y);
\draw[-arcsq] (U) -- (X);
\draw[-arcsq] (U) -- (Y);
\draw[-arcsq, red!90!black, line width=1pt] (U) -- (Z2);

\draw (2.0, -0.8) node {\small (c) Invalid IV set violating $\mathcal{C}3$};
\end{tikzpicture}

%%%%%%%%%%%%%%
%% Caption
%%%%%%%%%%%%%%
\vspace{-1mm}
\caption{Graphical illustration of IV set $\mathbf{Z} = \{Z_1, Z_2\}$ models,
where $\mathbf{U}$ denotes unmeasured confounders.
(a) $\mathbf{Z}$ is a valid IV set.
(b) $\mathbf{Z}$ is an invalid IV set, violating Condition $\mathcal{C}2$ because of the edge $Z_1 \to Y$.
(c) $\mathbf{Z}$ is an invalid IV set, violating Condition $\mathcal{C}3$ because of the edge $\mathbf{U} \to Z_2$.}
\label{Fig-IV-set}
\vspace{-3mm}
\end{figure*}
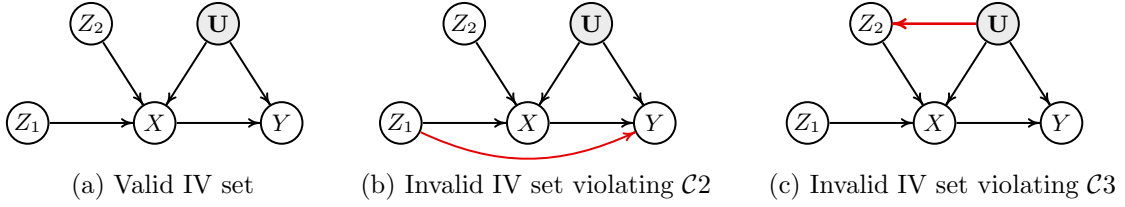

We next introduce the formal definitions of an IV and an IV set, which will be used throughout the paper.
\begin{Definition}[\textbf{Instrumental Variable (IV)}~\citep{pearl2009causality}]\label{Def-IV}
A variable $Z_i$ is said to be an IV w.r.t. the causal relation $X \to Y$ if it satisfies the following three conditions:
\begin{description}[leftmargin=15pt,itemsep=0pt,topsep=0pt,parsep=0pt]
    \item[$\mathcal{C}1$. (Relevance)] $Z_i$ is associated with the treatment $X$, i.e., $Z_i \nCI X$.
    \item[$\mathcal{C}2$. (Exclusion Restriction)] $Z_i$ does not directly affect the outcome $Y$, i.e., $Z_i \CI Y \mid \{X, \mathbf{U}\}$.
    \item[$\mathcal{C}3$. (Exogeneity)] $Z_i$ is independent of the unmeasured confounders $\mathbf{U}$, i.e., $Z_i \CI \mathbf{U}$.
\end{description}
\end{Definition}

\begin{Definition}[\textbf{IV Set}]\label{Def-IV-Set}
A set of variables $\mathbf{Z}$ is said to be a valid IV set \wrt $X \to Y$ if every nonempty subset $\mathbf{Z}' \subseteq \mathbf{Z}$ satisfies Conditions $\mathcal{C}1$--$\mathcal{C}3$ of Definition~\ref{Def-IV}, with $Z_i$ replaced by $\mathbf{Z}'$. Otherwise, $\mathbf{Z}$ is referred to as an invalid IV set \wrt $X \to Y$. 
\end{Definition}

As illustrated in Figure~\ref{Fig-IV-set}, Figure~\ref{Fig-IV-set}(a) shows a valid IV set $\{Z_1,Z_2\}$, where both variables satisfy Conditions $\mathcal{C}1$--$\mathcal{C}3$. 
By contrast, Figure~\ref{Fig-IV-set}(b) shows an invalid IV set in which $Z_1$ violates the exclusion restriction ($\mathcal{C}2$) due to the direct edge $Z_1 \to Y$, while Figure~\ref{Fig-IV-set}(c) shows another invalid IV set in which $Z_2$ violates exogeneity ($\mathcal{C}3$) due to its dependence on the unmeasured confounders $\mathbf{U}$.

\subsection{Causal Additive Models with Non-Constant Effects}\label{Sub-Section-Problem-Definition}

In this paper, we focus on \textbf{C}ausal \textbf{A}dditive \textbf{M}odels with \textbf{N}on-\textbf{C}onstant \textbf{E}ffects (CAM-NCE), involving the treatment $X$, the outcome $Y$, a candidate IV set $\mathbf{Z}$, and unmeasured confounders $\mathbf{U}$.
Without loss of generality, all variables are assumed to be mean-centered. We first present the covariate-free case for simplicity, and discuss the extension to baseline covariates in~\cref{Sec-CAT-Condition-Covariates}. Specifically, the data generation process under CAM-NCE can be expressed as:
\begin{equation}\label{Eq-Main-Model}
    \begin{aligned}
        X &= g_X(\mathbf{Z}) + \varphi_{X}(\mathbf{U}) + \varepsilon_{X},\\
        Y &= f(X) + g_Y(\widetilde{\mathbf{Z}}) + \varphi_{Y}(\mathbf{U}) + \varepsilon_{Y}, 
    \end{aligned}
\end{equation}
where the function $f(\cdot)$ represents the true but unknown causal effect of interest. 
The functions $f(\cdot)$, $g_{*}(\cdot)$, and $\varphi_{*}(\cdot)$ are assumed to be smooth functions defined on appropriate domains, i.e., $f:\mathbb{R}\!\to\!\mathbb{R}$, $g_X:\mathbb{R}^{|\mathbf{Z}|}\!\to\!\mathbb{R}$, $g_Y:\mathbb{R}^{|\widetilde{\mathbf{Z}}|}\!\to\!\mathbb{R}$, and $\varphi_{*}:\mathbb{R}^{|\mathbf{U}|}\!\to\!\mathbb{R}$. 
The noise terms $\varepsilon_X$ and $\varepsilon_Y$ are assumed to be mutually independent. The nonzero function $g_Y(\widetilde{\mathbf{Z}})$ indicates that the subset $\widetilde{\mathbf{Z}}\!\subseteq\!\mathbf{Z}$ directly affects the outcome $Y$, thereby violating the \emph{exclusion restriction} condition ($\mathcal{C}2$). 
Furthermore, if there exists any $Z_i \!\in\! \mathbf{Z}$ that is statistically dependent on $\mathbf{U}$, this implies a violation of the \emph{exogeneity} condition ($\mathcal{C}3$).

Note that, even given a valid IV, identifying the nonparametric structural function $f(\cdot)$ generally requires additional conditions. 
Following the standard additive nonparametric IV literature, we assume that a solution exists and impose the following completeness condition to ensure uniqueness given a valid IV~\citep{newey2003instrumental,ai2003efficient,darolles2011nonparametric,chernozhukov2007instrumental,newey2013nonparametric,singh2019kernel,bennett2019deep}.

\begin{Assumption}
[\textbf{Completeness Condition}]\label{Ass-completeness}
Given a valid IV $Z_i \in \mathbf{Z}$, for any measurable function $\psi(X)$ satisfying $\mathbb{E}[|\psi(X)|] < +\infty$, $\mathbb{E}[\psi(X)|Z_i]=0$ almost surely if and only if $\psi(X)=0$ almost surely. 
\end{Assumption}

The completeness condition is \emph{generic}, in the sense that it holds for ``most'' $s(X | Z_i)$—where $s(x | z_i)$ denotes the conditional density of $X$ given $Z_i$—if it holds for one~\citep{newey2013nonparametric,andrews2017examples}.
In particular, the finite-support case and models belonging to exponential families (such as Gaussian, Poisson, Binomial, or certain multivariate extensions) are known to satisfy completeness~\citep{newey2003instrumental,hu2018nonparametric}. Further sufficient conditions for completeness have been established in~\cite{d2011completeness,andrews2017examples}; and \cite{ hu2018nonparametric}.

\begin{Remark}
We highlight the following points regarding the proposed CAM-NCE model:
\begin{enumerate}[leftmargin=12pt, itemsep=2pt, topsep=2pt]
    \item \textbf{Relation to constant-effect models.} 
    When the functions $f(\cdot)$, $g_{*}(\cdot)$, and $\varphi_{*}(\cdot)$ are linear, the above model reduces to the \emph{Causal Additive Model with Constant Effects} (CAM-CE), such as the \emph{Additive Linear, Constant-Effects Model} (ALICE)~\citep{han2008detecting,kang2016instrumental,bowden2016consistent,windmeijer2019use,hartford2021valid,hartwig2017robust,guo2018confidence,windmeijer2021confidence,lin2024instrumental,bowden2015mendelian,kolesar2015identification,sanderson2022mendelian,silva2017learning,cheng2023discovering}. 
    In this work, we investigate a more challenging and general setting where $f(\cdot)$, $g_{*}(\cdot)$, and $\varphi_{*}(\cdot)$ can be nonlinear. 

    \item \textbf{Conditions to be assessed.} 
    Condition $\mathcal{C}1$ (relevance) can be readily assessed through standard statistical dependence tests. 
    Hence, our primary focus lies in addressing the remaining Conditions $\mathcal{C}2$ and $\mathcal{C}3$. 
    Notably, causal discovery methods based on conditional independence tests, such as the FCI (Fast Causal Inference) algorithm~\citep{spirtes1995causal} and its extensions~\citep{colombo2012learning,akbari2021recursive}, often yield a fully connected subgraph over $\{X, Y, Z_i\}$ in the presence of unmeasured confounders $\mathbf{U}$ (since $Z_i \nCI Y \mid X$). 
    As a result, verifying Conditions $\mathcal{C}2$--$\mathcal{C}3$ and identifying valid IVs remain challenging tasks.

    \item \textbf{Explicit representation of confounders.} 
    Several previous studies implicitly absorb unmeasured confounding into composite noise terms $\delta = \varphi_{X}(\mathbf{U}) + \varepsilon_{X}$ and $\epsilon = \varphi_{Y}(\mathbf{U}) + \varepsilon_{Y}$, which are typically correlated (e.g.,~\cite{newey2003instrumental,guo2016control,bennett2019deep,singh2019kernel}). In contrast, we explicitly represent $\mathbf{U}$ to facilitate subsequent theoretical derivations and analysis.
\end{enumerate}
\end{Remark}

\noindent\textbf{Our Goal.}  
The goal of this work is to develop a data-driven framework for testing the validity of IV sets under CAM-NCE, where Assumption~\ref{Ass-completeness} holds. 
Specifically, given a candidate IV set $\mathbf{Z}$ for the causal relation $X\to Y$, we aim to assess whether the variables in $\mathbf{Z}$ satisfy Conditions $\mathcal{C}1$--$\mathcal{C}3$.

\section{CAT Condition for Testing IV Set Validity}\label{Sec-CAT-Condition}

In this section, we introduce a testable criterion for assessing the validity of IV sets under CAM-NCE, termed the \emph{Cross Auxiliary-based Independence Test} (CAT) condition. 
We first show that, under the \emph{completeness condition} (Assumption~\ref{Ass-completeness}), the CAT condition is necessary for IV set validity. 
We then present a counterexample showing that completeness alone is insufficient to rule out all invalid IV sets. 
Finally, by imposing an additional \emph{cross distributional non-degeneracy condition} (Assumption~\ref{Ass-higher-order-condition}), we establish that the CAT condition becomes both necessary and sufficient for characterizing valid IV sets.

\subsection{CAT Condition: Definition and Necessity}\label{Subsection-CAT}
% \vspace{-2mm}
We first introduce the key concept of the \emph{auxiliary variable} together with the definition of the CAT condition for an IV set, which characterizes cross-independence relationships between the auxiliary variable and another distinct IV.
\begin{Definition}[\textbf{Auxiliary Variable}]\label{Def-auxiliary-variable} 
Let $X$, $Y$, and $Z_i \in \mathbf{Z}$ denote the treatment, outcome, and candidate IV, respectively. The auxiliary variable for the causal relationship  $X \to Y$ relative to $Z_i$ is defined as
% \vspace{-2mm}
\begin{align}
\label{Eq-auxiliary-variable}
    \mathcal{V}_{X \to Y||Z_i} \coloneqq Y - h_i(X),
\end{align}
% \vspace{-2mm}
where $h_i(\cdot)$ is a nonzero function satisfying $\mathbb{E}[ \mathcal{V}_{X \to Y||Z_i}| Z_i] = {0}$. 
\end{Definition} 

The notion of an auxiliary variable, or related residual-type constructions, has been used in various contexts~\citep{drton2004iterative,chen2017identification,cai2019triad,xie2024generalized,guo2024testability}.  
Different from these works, we use auxiliary variables to construct cross-independence relations among candidate IVs. 
To the best of our knowledge, such \emph{cross-independence}  relations have not been used to assess IV set validity under CAM-NCE.

\begin{Definition}[\textbf{CAT Condition}]\label{Definition-CAT-condition-set}
Let $X$, $Y$, and $\mathbf{Z}^\prime  \subseteq \mathbf{Z}$ denote the treatment, outcome, and a candidate IV set, respectively.
We say that $\{X,Y \|\ \mathbf{Z}^\prime\}$ satisfies the \textbf{C}ross \textbf{A}uxiliary-based Independence \textbf{T}est (CAT) condition if and only if, for every pair of distinct IVs $\{Z_i, Z_j\} \subseteq \mathbf{Z}^\prime$, the following cross-independence relationships hold:
\begin{equation}\label{Eq-CAT-Condition}
    \begin{aligned}
        \mathcal{V}_{X \to Y||Z_i}  \CI Z_j, \text{ and  } \mathcal{V}_{X \to Y||Z_j} \CI Z_i.
    \end{aligned}
\end{equation}
\end{Definition} 

Intuitively, the CAT condition performs a cross-check among candidate IVs: the auxiliary variable constructed using one candidate IV is tested for independence from another candidate IV in the set. 
For a pair $\{Z_i,Z_j\}$, CAT requires both $\mathcal{V}_{X \to Y\|Z_i} \CI Z_j$ and $\mathcal{V}_{X \to Y\|Z_j} \CI Z_i$ to hold.

We next provide a simple example to illustrate how the CAT condition detects violations of IV validity.

\begin{Example-set}[\textbf{Intuitive Example of the CAT Condition}]\label{Example-IV-set}
Consider the two causal graphs shown in Figure~\ref{Fig-IV-set}(a) and Figure~\ref{Fig-IV-set}(b).
Suppose that the corresponding data-generating mechanisms are given as follows:
\begin{itemize}[leftmargin=10pt]
    \item \textbf{Figure~\ref{Fig-IV-set}(a):}  
    $U = \varepsilon_U$, $Z_1 = \varepsilon_{Z_1}$, $Z_2 = \varepsilon_{Z_2}$, $X = 2 Z_1 +  Z_2 + 1.5 U + \varepsilon_X$, and 
    $Y = \exp(X) + 0.5 U + \varepsilon_Y.$ %\log(X)
    \item \textbf{Figure~\ref{Fig-IV-set}(b):}  
    $U = \varepsilon_U$,  $Z_1 = \varepsilon_{Z_1}$,  $Z_2 = \varepsilon_{Z_2}$, $X = 2 Z_1 +  Z_2 + 1.5 U + \varepsilon_X$, and  $Y = \exp(X) + 5 Z_1 + 0.5 U + \varepsilon_Y$. %\log(X)
\end{itemize}

\begin{figure}[ht]
    \centering
    \includegraphics[width=0.7\linewidth]{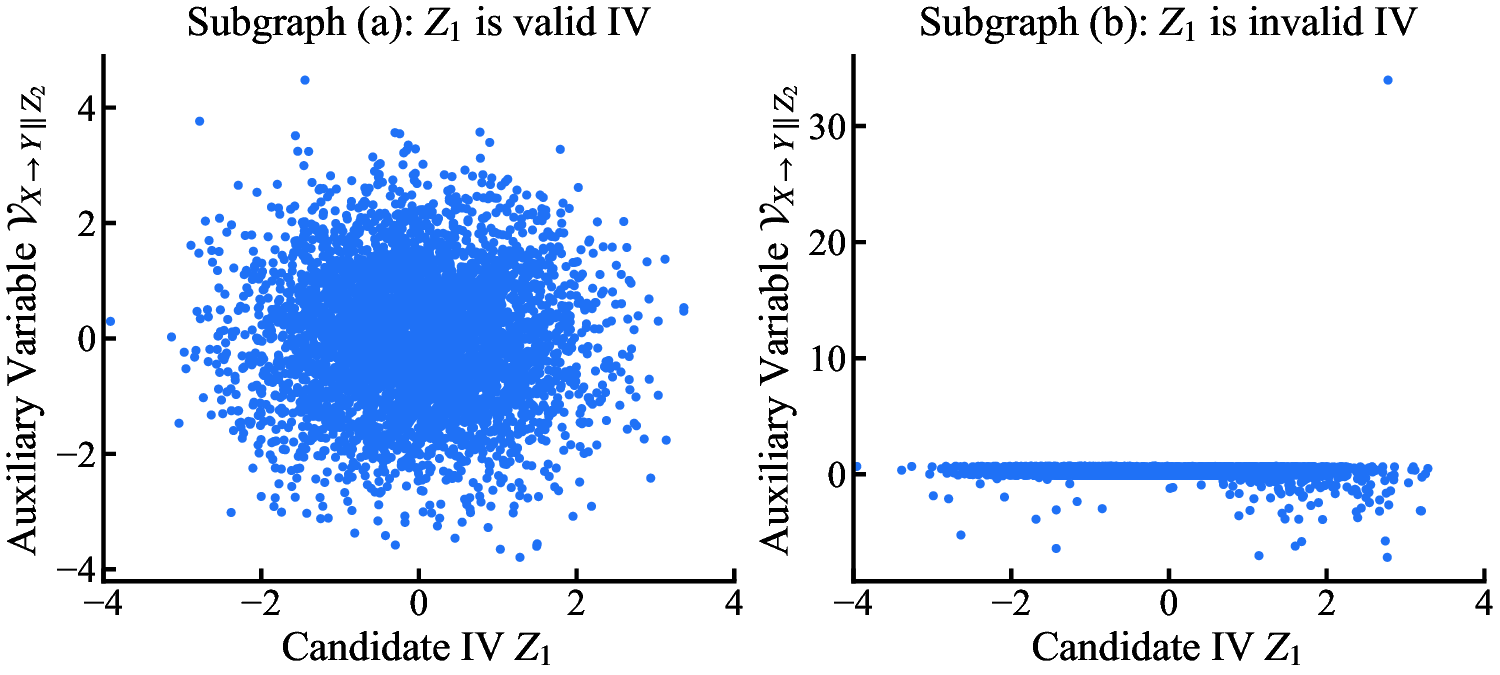}
    \caption{Scatter plots of candidate IV $Z_1$ and auxiliary variable $\mathcal{V}_{X \to Y \| Z_2}$ for Subgraphs (a) and (b) in Example \ref{Example-IV-set} when all noise terms
follow a standard Gaussian distribution.}
    \label{fig:example_nonconsatnt}
\end{figure}

Assume that all noise terms are mutually independent and follow a standard Gaussian distribution. Let $\mathbf{Z}^\prime = \{Z_1, Z_2\}$.

In Figure~\ref{Fig-IV-set}(a), both $Z_1$ and $Z_2$ are valid IVs. 
For either reference IV, the auxiliary variable removes the causal contribution $\exp(X)$ and becomes
\[
\mathcal{V}_{X \to Y\|Z_1}
=
\mathcal{V}_{X \to Y\|Z_2}
=
Y-\exp(X)
=
0.5U+\varepsilon_Y.
\]
Since $Z_1$, $Z_2$, and $U$ are mutually independent, we have 
$\mathcal{V}_{X \to Y\|Z_1}\CI Z_2$ and 
$\mathcal{V}_{X \to Y\|Z_2}\CI Z_1$. Hence, $\mathbf{Z}'$ satisfies the CAT condition. 
This is visualized in Figure~\ref{fig:example_nonconsatnt}, where $Z_1$ and $\mathcal{V}_{X \to Y\|Z_2}$ show no apparent dependence in the valid-IV-set case.

In Figure~\ref{Fig-IV-set}(b), however, $Z_1$ directly affects $Y$ and violates the exclusion restriction, while $Z_2$ remains valid. 
Using $Z_2$ as the reference IV, we obtain
\[
\mathcal{V}_{X \to Y\|Z_2}
=
Y-\exp(X)
=
5Z_1+0.5U+\varepsilon_Y,
\]
which is dependent on $Z_1$. 
Thus, $\mathcal{V}_{X \to Y\|Z_2}\nCI Z_1$, and the CAT condition is violated. 
As shown in Figure~\ref{fig:example_nonconsatnt}, this invalid-IV-set case exhibits a visible dependence pattern between $Z_1$ and $\mathcal{V}_{X \to Y\|Z_2}$, induced by the direct effect $Z_1\to Y$.

\end{Example-set}

It is worth noting that the validity of a single IV, such as $Z_1$ in Figure~\ref{Fig-IV-set}(b), cannot in general be verified from the observed distribution of $(X,Y,Z_1)$ alone. 
Indeed, the IV assumptions impose no testable constraints on this joint distribution, and different causal structures may induce the same observational distribution (see Section~3 of~\cite{chu2001semi}, and Proposition~3 of ~\cite{guo2024testability}). 
This motivates the need for testable criteria that exploit relations among multiple candidate IVs.

Building on the above intuition, we next establish that the CAT condition provides a necessary condition for IV set validity under CAM-NCE.

\begin{Theorem}[\textbf{{Necessary Condition for IV Set Validity}}]\label{Theorem-Necessary-Condition-IV-set}
Let $X$, $Y$, and $\mathbf{Z}^\prime  \subseteq \mathbf{Z}$ be the treatment, outcome, and candidate IV set in CAM-NCE, respectively. Suppose that $X$, $Y$, and $\mathbf{Z}^\prime$ are statistically dependent, and that Assumption~\ref{Ass-completeness} holds. If the candidate set $\mathbf{Z}^\prime$ is a valid IV set \wrt $X \to Y$, then $\{X, Y||\mathbf{Z}^\prime\}$ satisfies the CAT condition. 
\end{Theorem}

\begin{proof}[Proof sketch]
Since $\mathbf Z'$ is a valid IV set, each $Z_i\in\mathbf Z'$ is a valid IV. 
Under Assumption~\ref{Ass-completeness}, the conditional moment restriction 
$\mathbb E[Y-h_i(X)\mid Z_i]=0$ admits a unique solution.
By the standard identification result in nonparametric IV models~\citep{newey2003instrumental,newey2013nonparametric,d2011completeness,hu2018nonparametric}, this solution coincides with the true causal effect function $f(\cdot)$, \ie, $h_i(X)=f(X)$. Hence, for any pair of distinct IVs $\{Z_i,Z_j\}\subseteq\mathbf Z'$, the auxiliary variables satisfy 
\[
\mathcal V_{X\to Y\|Z_i}
=
\mathcal V_{X\to Y\|Z_j}
=
g_Y(\widetilde{\mathbf Z})+\varphi_Y(\mathbf U)+\varepsilon_Y. 
\]
Here, because $\mathbf Z'$ is a valid IV set, $Z_i$ and $Z_j$ do not directly affect $Y$ and are independent of $\mathbf U$. Moreover, under CAM-NCE, they are also independent of the remaining variables and noise terms involved in $g_Y(\widetilde{\mathbf Z})+\varepsilon_Y$.
Using the fact that measurable functions of disjoint subsets of mutually independent random variables are independent (see Theorem~\ref{The_function_indep} and Lemma~\ref{lemma_multi_function_indep} in Appendix), we have  
\[
    \mathcal V_{X\to Y\|Z_i}\CI Z_j,
    \quad
    \mathcal V_{X\to Y\|Z_j}\CI Z_i .
\]
Thus, every pair of distinct IVs in $\mathbf Z'$ satisfies the CAT condition. Consequently, $\{X, Y||\mathbf{Z}'\}$ satisfies the CAT condition. 
The full proof is provided in Appendix~\ref{proof-theroem-necessary}. 
\end{proof}

Theorem \ref{Theorem-Necessary-Condition-IV-set} states that if $\{X, Y||\mathbf{Z}^\prime\}$ violates the CAT condition, then the candidate IV set $\mathbf{Z}^\prime$ \wrt $X \to Y$ is invalid. Otherwise, $\mathbf{Z}^\prime$ may or may not be valid.

% \subsection{Theoretical Implications under the CAM-NCE Framework}
\subsection{Sufficient Condition for Characterizing IV Set Validity}

In the previous section, we have shown that the CAT condition is necessary for IV set validity. 
A natural question is \emph{whether, under Assumption~\ref{Ass-completeness}, every invalid IV set necessarily violates the CAT condition}. Unfortunately, the answer is negative. To clarify this issue, we present a counterexample below.

\begin{Example-set}\label{Example-same-coefficient-A2}(\textbf{Counterexample})
Let $Z_1=\varepsilon_{Z_1}$, $Z_2 = \varepsilon_{Z_2}$, $X = g_{X1}(Z_1) + g_{X2}(Z_2) + \varphi_X(\mathbf{U}) + \varepsilon_X$ and $Y = f(X) + g_{Y1}(Z_1) + g_{Y2}(Z_2) + \varphi_Y(\mathbf{U}) + \varepsilon_Y$, with $g_{Y1}(Z_1) = a \cdot g_{X1}(Z_1) + b_1$, $g_{Y2}(Z_2) = a \cdot g_{X2}(Z_2) + b_2$, where constant $a \neq 0$. In this model, both $Z_1$ and $Z_2$ directly affect $Y$ and hence violate the exclusion restriction. However, the same observational distribution can be generated by an alternative causal structure in which $\{Z_1,Z_2\}$ forms a valid IV set. 
Specifically, define $f^{\prime}(X) = f(X) + aX$, $Z_1^{\prime} = Z_1$, $Z_2^{\prime} = Z_2$, $X^{\prime} = X$, and $Y^{\prime} = f^{\prime}( X^{\prime}) + \varphi_Y(\mathbf{U}) + \varepsilon_{Y} - a\cdot \{\varphi_X(\mathbf{U}) + \varepsilon_X\} + b_1 + b_2$. 
Then, we have 
\begin{equation}\nonumber
\begin{aligned}
Y'
&= f(X)+aX+\varphi_Y(\mathbf U)+\varepsilon_Y
-a\{\varphi_X(\mathbf U)+\varepsilon_X\}
+b_1+b_2  \\
&= f(X)+a\{g_{X1}(Z_1)+g_{X2}(Z_2)\}
+\varphi_Y(\mathbf U)+\varepsilon_Y
+b_1+b_2  \\
&= f(X)+g_{Y1}(Z_1)+g_{Y2}(Z_2)
+\varphi_Y(\mathbf U)+\varepsilon_Y
=Y .
\end{aligned}
\end{equation}
Thus, $(X,Y,Z_1,Z_2)$ and $(X',Y',Z_1',Z_2')$ have the same observational distribution, although $\{Z_1,Z_2\}$ is invalid in the original structure but valid in the alternative structure. 
This shows that Assumption~\ref{Ass-completeness} alone cannot distinguish all invalid IV sets.
\end{Example-set}

The above example shows that Assumption~\ref{Ass-completeness} alone does not guarantee that the CAT condition can rule out all invalid IV sets. To obtain a sufficient condition, we introduce an additional non-degeneracy condition that excludes those cases and makes IV set validity testable under the CAM-NCE.

\begin{Assumption}[\textbf{Cross Distributional Non-degeneracy Condition}]\label{Ass-higher-order-condition} 
For any invalid candidate IV set $\mathbf{Z}'\subseteq\mathbf{Z}$ with $|\mathbf{Z}'|\geq 2$, there exists a pair of distinct candidate IVs $\{Z_i,Z_j\}\subseteq\mathbf{Z}'$ such that the joint densities $p(\mathcal{V}_{X\to Y\|Z_i},Z_j)$ and $p(\mathcal{V}_{X\to Y\|Z_j},Z_i)$ are twice continuously differentiable. Moreover, at least one of the following cross second-order partial derivatives,
$$\frac{\partial^2 \operatorname{log}p(\mathcal{V}_{X \to Y||Z_i},Z_j)}{\partial \mathcal{V}_{X \to Y||Z_i}\partial Z_j} \text{ or } \frac{\partial^2 \operatorname{log}p(\mathcal{V}_{X \to Y||Z_j},Z_i)}{\partial \mathcal{V}_{X \to Y||Z_j}\partial Z_i}$$ is non-zero on a set with non-zero Lebesgue measure, where $\mathcal{V}_{X \to Y||Z_i} = g_Y(\widetilde{\mathbf{Z}}) + \varphi_{Y}(\mathbf{U}) + \varepsilon_Y -f_{bias}^i (X)$, and $f_{bias}^i (X)=h_i(X)-f(X)$.
\end{Assumption}

Assumption \ref{Ass-higher-order-condition} is a natural condition that one expects to hold to identify the invalid IV set. 
Its intuition comes from the linear separability of the logarithm of the joint density of independent variables. 
Specifically, for a set of independent random variables with a twice-differentiable joint density, the Hessian matrix of the log-density is diagonal~\citep{lin1997factorizing}. 
Applying this property to $\mathcal{V}_{X \to Y\|Z_i}$ and $Z_j$, if 
$\mathcal{V}_{X \to Y\|Z_i}\CI Z_j$, then their joint density factorizes as 
$p(\mathcal{V}_{X \to Y \| Z_i}, Z_j) = p(\mathcal{V}_{X \to Y \| Z_i}) \cdot p(Z_j)$. 
Equivalently, the log-density is additively separable, and hence the corresponding off-diagonal Hessian entry vanishes: 
$\frac{\partial^2 \log p(\mathcal{V}_{X \to Y \| Z_i}, Z_j)}{\partial \mathcal{V}_{X \to Y \| Z_i} \, \partial Z_j} = 0$. 
The same argument applies after exchanging $Z_i$ and $Z_j$. 
Assumption~\ref{Ass-higher-order-condition} requires that, for an invalid IV set, such cross derivatives do not vanish for at least one cross pair, thereby making the violation detectable through the CAT condition.

\begin{Remark} 
Assumption~\ref{Ass-higher-order-condition}, referred to as the Cross Distributional Non-degeneracy Condition, can be viewed as a distributional analogue of the Algebraic Equation Condition in~\cite{guo2025data}. 
While Guo et al.'s condition is derived from an explicit algebraic expansion under constant-effect models, our condition is formulated through cross second-order derivatives of joint log-densities and does not require an invertible transformation between observed variables and latent noise terms. 
This makes the proposed condition applicable to the more general CAM-NCE framework.
\end{Remark}

The following example gives a concrete illustration of how a violation of IV validity can induce a nonzero cross second-order derivative in the joint density. For clarity, we present the calculation in a constant-effect setting, where closed-form expressions are available. The example is intended to illustrate the intuition behind the more general non-constant-effect case.
 
\begin{Example-set}[\textbf{Illustration of the Cross Distributional Non-degeneracy Condition}]
\label{ex:invalid-IV-nonzero-derivative}
Consider the following data-generating process: $U=\varepsilon_U$, $Z_1=\varepsilon_{Z_1}$, $Z_2=\varepsilon_{Z_2}$, $X=Z_1+Z_2+Z_1^2+Z_2^2+U+\varepsilon_X$, and $Y=X+2Z_1Z_2+U+\varepsilon_Y$, where $\varepsilon_{Z_1},\varepsilon_{Z_2},\varepsilon_U, \varepsilon_X,\varepsilon_Y
\stackrel{\mathrm{ind}}{\sim} N(0,1)$. Here, the candidate IVs \(Z_1\) and \(Z_2\) are invalid because they directly affect \(Y\) through the term \(2Z_1Z_2\), thereby violating the exclusion restriction. 

For $Z_1$, by the definition of the auxiliary variable, $\mathcal{V}_1=Y-\beta_1 X=2Z_1Z_2+U+\varepsilon_Y$, where $\beta_1= \frac{\operatorname{Cov}(Y,Z_1)}{\operatorname{Cov}(X,Z_1)} = 1$. The joint density of \((\mathcal{V}_1,Z_2)\) can be factorized as $p(\mathcal{V}_1,Z_2)=p(\mathcal{V}_1\mid Z_2)p(Z_2)$, where \(Z_2\sim N(0,1)\). 
Conditional on $Z_2=z_2$, since $Z_1$, $U$, and $\varepsilon_Y$ are mutually independent standard normal variables, we have $\mathcal{V}_1\mid Z_2=z_2\sim N(0,4z_2^2+2)$. Therefore, $\log p(\mathcal{V}_1,z_2)=\log p(\mathcal{V}_1|z_2)p(z_2)=-\log(2\pi)-\frac{z_2^2}{2}-\frac{1}{2}\log(4z_2^2+2)-\frac{\mathcal{V}_1^2}{2(4z_2^2+2)}$. 
Then, $\frac{\partial^2 \log p(\mathcal{V}_1,z_2)}
    {\partial \mathcal{V}_1\,\partial z_2}
    =
    \frac{8\mathcal{V}_1z_2}{(4z_2^2+2)^2}$. Since this quantity is nonzero whenever \(\mathcal{V}_1z_2\neq 0\), the cross second-order partial derivative is nonzero on a set with nonzero Lebesgue measure. Therefore, the candidate IV set $\{Z_1,Z_2\}$ satisfies the cross distributional non-degeneracy condition (Assumption \ref{Ass-higher-order-condition}). 
\end{Example-set}

To better understand the Cross Distributional Non-degeneracy Condition 
(Assumption~\ref{Ass-higher-order-condition}), we next provide a sufficient condition under which this assumption fails.

\begin{Proposition}[\textbf{Sufficient Condition for Violation of Assumption \ref{Ass-higher-order-condition}}]
\label{Pro-violate-assumption2}
Let $X$, $Y$, and ${\mathbf{Z}}'\subseteq\mathbf{Z}$ be the treatment, outcome, and a candidate IV set in CAM-NCE, respectively. 
Suppose that each $Z_i\in{\mathbf{Z}}'$ is relevant and exogenous but violates only the exclusion restriction. 
Assume further that, for $\mathbf{Z}'$, its direct causal effect on $Y$ is an affine transformation of its direct causal effect on $X$, namely, $g_{Y}({\mathbf{Z}}')=a\cdot g_{X}({\mathbf{Z}}')+b$, where $a$ is a nonzero constant, and $b$ is a constant. 
Then Assumption~\ref{Ass-higher-order-condition} fails for the invalid IV set ${\mathbf{Z}}'$. 
This special case is illustrated graphically in Figure~\ref{Fig-invalid-IV-set-A2}. 
\end{Proposition}

\begin{figure}[h]
\centering
\begin{tikzpicture}[
    scale=0.9,
    line width=0.8pt,
    inner sep=0.4mm,
    shorten >=.1pt,
    shorten <=.1pt
]
\tikzstyle{every node}+=[inner sep=0pt]

% nodes
\draw (3.5,1.5) node(U)
    [circle, draw=black, fill=gray!15, minimum size=0.6cm]
    {\footnotesize $\mathbf{U}$};

\draw (-0.3,0) node(Z)
    [circle, draw=black, fill=white, minimum size=0.6cm]
    {\footnotesize ${\mathbf Z}'$};

\node[
    rectangle,
    rounded corners=2pt,
    draw=black!70,
    fill=white,
    text width=1.0cm,
    align=center,
    minimum height=5mm
] (S) at (1.25,1.5)
{\footnotesize $\mathbf Z\setminus {\mathbf Z}'$};

\draw (2.3,0) node(D)
    [circle, draw=black, fill=white, minimum size=0.55cm]
    {\footnotesize $X$};

\draw (4.6,0) node(Y)
    [circle, draw=black, fill=white, minimum size=0.55cm]
    {\footnotesize $Y$};

% arrows
\draw[-arcsq] (U) -- (D);
\draw[-arcsq] (U) -- (Y);

\draw[-arcsq]
    (Z) -- (D)
    node[midway, above]
    {$g_X({\mathbf Z}')$};

\draw[-arcsq] (S) -- (D);

\draw[-arcsq] (D) -- (Y);

\draw[-arcsq, red!85!black, line width=1pt]
    (Z) edge[bend right=33]
    node[pos=0.5, above, text=red!85!black]
    {$g_Y({\mathbf Z}')$}
    (Y);

\end{tikzpicture}

\caption{Graphical illustration of an invalid IV set $\mathbf{Z}' \subseteq \mathbf{Z}$, where every IV in $\mathbf{Z}'$ violates only the exclusion restriction condition ($\mathcal{C}2$), while all variables in $\mathbf{Z}\setminus\mathbf{Z}'$ are valid IVs. For the invalid IV set $\mathbf{Z}'$, the direct effects satisfy $g_{Y}({\mathbf{Z}}')=a\cdot g_{X}({\mathbf{Z}}')+b$, where $a\neq0$.}
\label{Fig-invalid-IV-set-A2}
\end{figure}
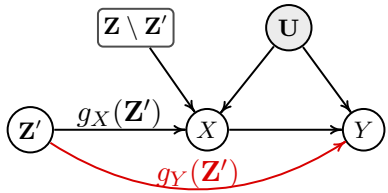

\begin{proof}
See Appendix \ref{proof-Coro-same-constant} for its proof. 
\end{proof}

As shown in Proposition~\ref{Pro-violate-assumption2}, when Assumption~\ref{Ass-higher-order-condition} is violated, certain IV sets may fail to satisfy the exclusion restriction, leading to non-identifiability of causal effects. We will now show that under additional Assumption~\ref{Ass-higher-order-condition}, IV sets can be uniquely identified.

\begin{Proposition}[\textbf{Sufficient Condition for IV Set Validity}]\label{Pro-Test-Higher-CAM-NCE}
Let $X$, $Y$, and $\mathbf{Z}^\prime \subseteq \mathbf{Z}$ be the treatment, outcome, and candidate IV set in CAM-NCE, respectively. Suppose that $X$, $Y$, and $\mathbf{Z}'$ are statistically dependent, and that Assumptions \ref{Ass-completeness} and \ref{Ass-higher-order-condition} hold. If the candidate IV set $\mathbf{Z}^\prime$ is invalid, then $\{X, Y||\mathbf{Z}^\prime\}$ violates the CAT condition. 
\end{Proposition}

\begin{proof}
See Appendix~\ref{prooof-Pro-Test-Higher-CAM-NCE} for its proof.
\end{proof}

Proposition~\ref{Pro-Test-Higher-CAM-NCE} shows that, under Assumptions~\ref{Ass-completeness} and~\ref{Ass-higher-order-condition}, every invalid IV set violates the CAT condition. 
Combining this result with Theorem~\ref{Theorem-Necessary-Condition-IV-set}, we obtain the following necessary and sufficient characterization of IV set validity under CAM-NCE.

\begin{Theorem}[\textbf{Necessary and Sufficient Condition for IV Set Validity}]\label{Theorem-necessary-sufficient-condition-CAM-NCE}
Let $X$, $Y$, and $\mathbf{Z}^\prime \subseteq \mathbf{Z}$ be the treatment, outcome, and candidate IV set in CAM-NCE, respectively. Suppose that $X$, $Y$, and $\mathbf{Z}'$ are statistically dependent, and that Assumptions \ref{Ass-completeness} and \ref{Ass-higher-order-condition} hold. The candidate IV set $\mathbf{Z}^\prime$ is a valid IV set relative to $X \to Y$ if and only if $\{X, Y||\mathbf{Z}^\prime\}$ satisfies the CAT condition. 
\end{Theorem}

\begin{proof}
See Appendix~\ref{proof-Theorem-necessary-sufficient-condition-CAM-NCE} for its proof.
\end{proof}

The theoretical implications of the CAT condition are summarized in Figure~\ref{fig:cat_theory_summary}.

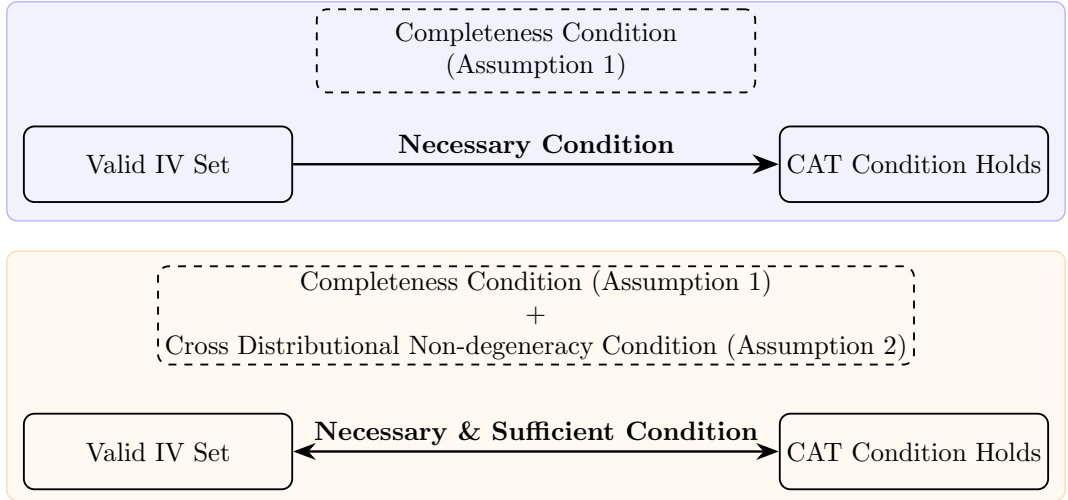
\begin{figure}[h]
\centering
\begin{tikzpicture}[
    node distance=4.5cm,
    box/.style={
        rectangle,
        rounded corners,
        draw=black,
        line width=0.7pt,
        minimum width=3.55cm,
        minimum height=1.0cm,
        align=center,
        font=\small
    },
    assump/.style={
        rectangle,
        rounded corners,
        draw=black,
        dashed,
        line width=0.7pt,
        minimum width=5.8cm,
        minimum height=1.1cm,
        align=center,
        font=\small
    },
    arrow/.style={
        -{Stealth[length=3.5mm]},
        line width=0.9pt
    },
    botharrow/.style={
        {Stealth[length=3.0mm]}-{Stealth[length=3.0mm]},
        line width=0.9pt
    },
    label/.style={
        font=\small\bfseries,
        align=center,
        yshift=2pt % 稍微向上抬高避免贴着箭头
    }
]

% Title
\node[font=\large\bfseries] at (0,2.8) {Theoretical Role of the CAT Condition};

% --- 第一组 (First Row) ---
\node[assump] (a1) at (0, 1.5) {Completeness Condition\\(Assumption~1)};
\node[box] (valid1) at (-5.0, 0) {Valid IV Set};
\node[box] (cat1) at (5.0, 0) {CAT Condition Holds};

\draw[arrow] (valid1) -- node[above, label] {Necessary Condition} (cat1);

% --- 第二组 (Second Row) ---
\node[assump, minimum width=10.0cm, minimum height=1.3cm] (a13) at (0, -2.0) {
    Completeness Condition (Assumption~1) \\ $+$\\
    Cross Distributional Non-degeneracy Condition (Assumption~2)
};

\node[box] (valid2) at (-5.0, -3.8) {Valid IV Set};
\node[box] (cat2) at (5.0, -3.8) {CAT Condition Holds};

\draw[botharrow] (valid2) -- node[above, label] {Necessary \& Sufficient Condition} (cat2);

% --- 背景图层 (Background layers) ---
\begin{scope}[on background layer]
    % one
    \node[
        rounded corners,
        fill=blue!5,
        draw=blue!30,
        line width=0.5pt,
        minimum width=14.0cm,
        minimum height=2.9cm,
        inner sep=0pt
    ] at (0.0, 0.7) {};

    % two
    \node[
        rounded corners,
        fill=orange!5,
        draw=orange!30,
        line width=0.5pt,
        minimum width=14.0cm,
        minimum height=3.3cm,
        inner sep=0pt
    ] at (0.0, -2.8) {};
\end{scope}

\end{tikzpicture}
\caption{Flowchart of the theoretical implications of the CAT condition.}
\label{fig:cat_theory_summary}
\end{figure}

As shown in Figure~\ref{fig:cat_theory_summary}, the CAT condition is necessary for IV set validity under the completeness condition, and becomes both necessary and sufficient when the Cross Distributional Non-degeneracy Condition is additionally imposed.

\section{Practical Testing Algorithm from Data}\label{Sec-CAT Algorithm}
In this section, we first extend the CAT condition to settings with baseline covariates, where IV validity is assessed after adjusting for observed covariates. 
We then develop a finite-sample testing algorithm for assessing the validity of candidate IV sets from data. 
\vspace{-2mm}

\subsection{CAT Condition with Covariates}\label{Sec-CAT-Condition-Covariates}
In practice, observed covariates such as age, gender, and other background variables may affect the treatment, outcome, and candidate IVs. It is therefore necessary to assess IV validity after adjusting for such covariates.

Under CAM-NCE with covariates, the data-generating process is given by 
\begin{equation}\label{Eq-Main-Model-covariates}
    \begin{aligned}
        X &= g_X(\mathbf{Z}) + t_X(\mathbf{W}) + {\varphi_{X}(\mathbf{U}) + \varepsilon_{X}},\\
        Y &= f(X) + t_Y(\mathbf{W}) + g_Y(\widetilde{\mathbf{Z}}) + {\varphi_{Y}(\mathbf{U}) + \varepsilon_{Y}}, 
    \end{aligned}
\end{equation} 
where $\mathbf{W}$ denotes the covariates.
Below, we show that the additive structure in Equation~\eqref{Eq-Main-Model-covariates} allows the CAT condition in Definition~\ref{Definition-CAT-condition-set} to be extended to covariate settings through regression adjustment.

\begin{Definition}[\textbf{CAT Condition with Covariates}]\label{Definition-CAT-condition-set-covariates} 
Let $X$, $Y$, $\mathbf{W}$, and $\mathbf{Z}^\prime \subseteq \mathbf{Z}$ denote the treatment, outcome, covariates, and candidate IV set, respectively. Furthermore, let $\mathcal{X}$, $\mathcal{Y}$, and $\boldsymbol{\mathcal{Z}}'$ denote the residuals obtained by regressing $X$, $Y$, and $\mathbf{Z}'$ on $\mathbf{W}$, respectively (e.g., for each IV, $\mathcal{Z}_i \coloneqq Z_i - \mathbb{E}[Z_i|\mathbf{W}]$).  
We say that $\{X, Y||(\mathbf{Z}^\prime, \mathbf{W})\}$ satisfies the CAT condition if and only if for every pair of distinct IVs $\{Z_i, Z_j\} \subseteq \mathbf{Z}^\prime$, the following cross-independence relationships hold: 
\begin{equation}
\begin{aligned}
\mathcal{V}_{\mathcal{X} \to \mathcal{Y}||\mathcal{Z}_i} \CI \mathcal{Z}_j, \text{ and  } \mathcal{V}_{\mathcal{X} \to \mathcal{Y}||\mathcal{Z}_j} \CI \mathcal{Z}_i.
\end{aligned}
\end{equation}
\end{Definition}

Based on Definition \ref{Definition-CAT-condition-set-covariates} and Theorem \ref{Theorem-Necessary-Condition-IV-set}, we derive the following necessary condition for IV set validity in the presence of covariates $\mathbf{W}$.

\begin{Corollary}[\textbf{{Necessary Condition for IV Set Validity with Covariates}}]\label{Corollary-Necessary-Condition-IV-set-covariates}
Let $X$, $Y$, $\mathbf{W}$, and $\mathbf{Z}' \subseteq \mathbf{Z}$ be the treatment, outcome, covariates, and candidate IV set in CAM-NCE, respectively. Suppose that $X$, $Y$, $\mathbf{W}$, and $\mathbf{Z}'$ are statistically dependent, and that Assumption \ref{Ass-completeness} holds for the residualized variables. If the candidate IV set $\mathbf{Z}'$ is a valid IV set \wrt $X \to Y$ given $\mathbf{W}$, then $\{X, Y||(\mathbf{Z}', \mathbf{W})\}$ satisfies the CAT condition.
\end{Corollary}

\begin{proof}
See Appendix~\ref{proof-Corollary-Necessary-Condition-IV-set-covariates} for its proof.  
\end{proof}

Corollary \ref{Corollary-Necessary-Condition-IV-set-covariates} states that if $\{X,Y||(\mathbf{Z}', \mathbf{W})\}$ violates the CAT condition, then $\mathbf{Z}'$ is an invalid IV set. Furthermore, to obtain a necessary and sufficient characterization in the
presence of covariates, we impose Assumption~\ref{Ass-higher-order-condition} on the covariate-adjusted residual variables.

\begin{Corollary}[\textbf{Necessary and Sufficient Condition for IV Set Validity with Covariates}]\label{Corollary-necessary-sufficient-condition-CAM-NCE-covariates}
Let $X$, $Y$, $\mathbf{W}$, and $\mathbf{Z}' \subseteq \mathbf{Z}$ be the treatment, outcome, covariates, and candidate IV set in a CAM-NCE, respectively. Suppose that $X$, $Y$, $\mathbf{W}$, and $\mathbf{Z}'$ are statistically dependent. Further suppose that Assumptions~\ref{Ass-completeness} and~\ref{Ass-higher-order-condition} hold for the covariate-adjusted residual variables. The candidate IV set $\mathbf{Z}^\prime$ is a valid IV set relative to $X \to Y$ given $\mathbf{W}$ if and only if $\{X, Y||(\mathbf{Z}^\prime, \mathbf{W})\}$ satisfies the CAT condition. 
\end{Corollary}

\begin{proof}
See Appendix~\ref{proof-Corollary-necessary-sufficient-condition-CAM-NCE-covariates} for its proof.
\end{proof}
% \vspace{-2mm}

\subsection{CAT Algorithm with Finite Samples}
In this subsection, we develop a practical CAT algorithm for finite-sample data. For generality, we present the algorithm in the setting with baseline covariates $\mathbf{W}$. 
Let $\widehat{\mathcal X}$, $\widehat{\mathcal Y}$, and 
$\widehat{\boldsymbol{\mathcal Z}}=\{\widehat{\mathcal Z}_1,\ldots,\widehat{\mathcal Z}_m\}$ denote the residuals obtained by regressing $X$, $Y$, and each candidate IV $Z_i\in\mathbf Z=\{Z_1,\ldots,Z_m\}$ on $\mathbf W$, respectively. 
When no covariates are available, this residualization step is omitted and the original variables are used directly.

Since the CAT results are given in terms of population-level  (Theorems~\ref{Theorem-Necessary-Condition-IV-set}$\sim$\ref{Theorem-necessary-sufficient-condition-CAM-NCE}
and Corollaries~\ref{Corollary-Necessary-Condition-IV-set-covariates}$\sim$\ref{Corollary-necessary-sufficient-condition-CAM-NCE-covariates}), implementing it with observational samples requires addressing three practical questions:
\begin{itemize}
    \item[$\mathcal{Q}1$.] How can one efficiently search for a valid IV set from a collection of candidate IVs?
    \item[$\mathcal{Q}2$.] How can one estimate the auxiliary variable $\mathcal{V}_{\mathcal{X} \to \mathcal{Y}\|\mathcal{Z}_i}$ for each candidate IV $Z_i$?
    \item[$\mathcal{Q}3$.] How can the CAT condition be implemented with finite samples?
\end{itemize}

We next address these questions in turn.

\noindent\textbf{$\mathcal{Q}1$: Searching for candidate IV sets.}
Since the CAT condition is defined through pairwise cross-independence relations, we focus on candidate subsets with at least two IVs. 
Searching over all such subsets of $\widehat{\boldsymbol{\mathcal Z}}$ requires examining $2^m-m-1$ possibilities, which can be computationally expensive when $m$ is large. 
To make the search tractable, we introduce a user-specified parameter $K\geq 2$, denoting the target size of the IV set to be selected. 
Given $K$, we restrict attention to candidate subsets of size $K$, resulting in $\binom{m}{K}$ subsets. 
Specifically, we consider all subsets $\mathcal S_c\subseteq \widehat{\boldsymbol{\mathcal Z}}$ with $|\mathcal S_c|=K$. 
When prior knowledge about $K$ is unavailable, we suggest evaluating $K$ over a range of plausible values, starting from small values, and using the CAT-based score below to select among the resulting candidate subsets.

\noindent\textbf{$\mathcal{Q}2$: Estimating auxiliary variables.}
To construct the auxiliary variable for each candidate IV, the key step is to estimate the function $h_i(\cdot)$ satisfying the conditional moment restriction $\mathbb{E}\!\left[\mathcal{Y}-h_i(\mathcal{X})\mid \mathcal{Z}_i\right]=0$. Under the additive model, if $\mathcal{Z}_i$ is a valid IV, then $h_i(\cdot)$ is identifiable under the completeness condition (Assumption~\ref{Ass-completeness}) and coincides with the structural response function of the treatment on the outcome~\citep{newey2003instrumental}. 
In this case, $h_i(\cdot)$ can be consistently estimated using suitable IV estimators. 
If $\mathcal{Z}_i$ is invalid, the resulting estimator may be biased, and this bias will be reflected in the corresponding auxiliary variable. Various estimators have been developed for additive nonparametric IV models, including sieve-based estimators~\citep{newey2013nonparametric}, kernel-based estimators~\citep{singh2019kernel,muandet2020dual}, deep IV estimators~\citep{hartford2017deep,lin2019one}, and control-function estimators~\citep{newey1999nonparametric,guo2016control}. 
In our implementation, we use the control-function IV estimator for non-constant effects~\citep{guo2016control}. 
When prior knowledge suggests a constant-effect model, we instead use the two-stage least squares estimator~\citep{wooldridge2010econometric}. 
Given the estimated function $\widehat h_i(\cdot)$, the empirical auxiliary variable is constructed as $\widehat{\mathcal{V}}_{\mathcal{X}\to\mathcal{Y}\|\mathcal{Z}_i}= \widehat{\mathcal Y}-\widehat h_i(\widehat{\mathcal X})$.

\noindent\textbf{$\mathcal{Q}3$: Implementing the CAT condition with finite samples.}
For each candidate subset $\mathcal{S}_c$, the CAT condition requires every pair of distinct candidate IVs to satisfy two directed cross-independence relations. 
Specifically, for each pair $\{\widehat{\mathcal{Z}}_i,\widehat{\mathcal{Z}}_j\}\subseteq \mathcal{S}_c$, we need to assess
\[
\widehat{\mathcal{V}}_{\mathcal{X}\to\mathcal{Y}\|\mathcal{Z}_i}
    \CI \widehat{\mathcal{Z}}_j,
    \quad
\widehat{\mathcal{V}}_{\mathcal{X}\to\mathcal{Y}\|\mathcal{Z}_j}
    \CI \widehat{\mathcal{Z}}_i.
\]
One could use formal independence tests for continuous variables and determine whether the CAT condition holds based on the corresponding $p$-values. 
However, since our goal is to compare many candidate subsets in finite samples, we use distance correlation as a numerical dependence score. 
According to~
\cite{szekely2007measuring,szekely2009brownian}, distance correlation is zero if and only if the two random vectors are independent. Hence, whenever the population-level CAT condition holds, the corresponding population distance correlation is zero. In finite samples, smaller sample distance correlations therefore provide stronger empirical support for the CAT condition. Let $\widehat{\operatorname{dCor}}(\cdot,\cdot)$ denote the sample distance correlation. For each candidate subset $\mathcal{S}_c$, we define its CAT score as the sum of directed pairwise distance correlations: 
\vspace{-2mm}
\begin{equation}\nonumber
    \begin{aligned}
\widehat{T}_{\mathcal{S}_c}
=\sum_{\substack{\{\widehat{\mathcal{Z}}_i,\widehat{\mathcal{Z}}_j\}\subseteq \mathcal{S}_c\\ i < j}}
\Bigg[
&\widehat{\operatorname{dCor}}\!\left(
\widehat{\mathcal{V}}_{\mathcal{X}\to\mathcal{Y}\|\mathcal{Z}_i},
\widehat{\mathcal{Z}}_j
\right) 
+
\widehat{\operatorname{dCor}}\!\left(
\widehat{\mathcal{V}}_{\mathcal{X}\to\mathcal{Y}\|\mathcal{Z}_j},
\widehat{\mathcal{Z}}_i
\right)
\Bigg].
\end{aligned}
\end{equation}
For a subset of size $K$, this score aggregates $K(K-1)/2$ unordered IV pairs and $K(K-1)$ directed distance-correlation terms. 
Finally, we select the candidate subset with the smallest empirical CAT score: 
\[
    \widehat{\mathcal S}
    =
    \arg\min_{\mathcal{S}_c\subseteq \widehat{\boldsymbol{\mathcal{Z}}},\ |\mathcal{S}_c|=K}
    \widehat{T}_{\mathcal{S}_c}.
\]
The selected subset $\widehat{\mathcal S}$ is therefore the candidate IV set that is most consistent with the CAT condition in finite samples.

\begin{algorithm}[h]
\caption{CAT}
\label{algorithm-CAT}
\begin{algorithmic}[1]
    \REQUIRE Observed dataset $\mathcal{D}=\{X,Y,\mathbf{W},\mathbf{Z}\}$; $K$, the number of valid IVs to select.
    \ENSURE Selected IV set $\widehat{\mathcal{S}}$ and its corresponding distance-correlation score $T_{\mathcal{S}}$. 

    \STATE \textbf{Initialize:} the selected IV set $\widehat{\mathcal{S}}\gets \emptyset$, the distance correlation matrix 
    $\mathbf{M}\gets \mathbf{0}_{m\times m}$; 

    \IF{$\mathbf{W}\neq \emptyset$}
        \STATE $\widehat{\mathcal{X}},\widehat{\mathcal{Y}},\widehat{\boldsymbol{\mathcal{Z}}}
        \gets$ residuals from the regressions of $X,Y,\mathbf{Z}$ on $\mathbf{W}$, respectively;
    \ELSE
        \STATE $\widehat{\mathcal{X}},\widehat{\mathcal{Y}},\widehat{\boldsymbol{\mathcal{Z}}}\gets X,Y,\mathbf{Z}$;
    \ENDIF

    % \STATE Let $\boldsymbol{\mathcal{Z}}=\{\mathcal{Z}_1,\ldots,\mathcal{Z}_m\}$;

    \FOR{$i=1,\ldots,m$}
        \IF{the non-constant causal effect of $X\to Y$ is adopted}
            \STATE $\widehat{h}_i(\widehat{\mathcal{X}})
            \gets \text{Control-Function IV Estimator}(\widehat{\mathcal{X}},\widehat{\mathcal{Y}},\widehat{\mathcal{Z}}_i)$; 
        \ELSE
            \STATE $\widehat{h}_i(\widehat{\mathcal{X}})
            \gets \widehat{\beta}_i\widehat{\mathcal{X}}$, where
            $\widehat{\beta}_i\gets
            \frac{\operatorname{Cov}(\widehat{\mathcal{Y}},\widehat{\mathcal{Z}}_i)}
            {\operatorname{Cov}(\widehat{\mathcal{X}},\widehat{\mathcal{Z}}_i)}$;
        \ENDIF
        \STATE $\widehat{\mathcal{V}}_i
        \gets
        \widehat{\mathcal{Y}}-\widehat{h}_i(\widehat{\mathcal{X}})$;
    \ENDFOR

    % \STATE \textbf{Initialize:} $\mathbf{M}\in\mathbb{R}^{m\times m}$ with all entries equal to $0$;

    \FOR{each unordered pair $\{\mathcal{Z}_i,\mathcal{Z}_j\}$ with $1\leq i<j\leq m$}
        \STATE
        $\widehat{\mathbf{M}}_{ij}
        \gets
        \widehat{\operatorname{dCor}}(\widehat{\mathcal{V}}_i,\widehat{\mathcal{Z}}_j)
        +
        \widehat{\operatorname{dCor}}(\widehat{\mathcal{V}}_j,\widehat{\mathcal{Z}}_i)$;
        \STATE $\widehat{\mathbf{M}}_{ji}\gets \widehat{\mathbf{M}}_{ij}$;
    \ENDFOR

    \STATE Let $\mathcal{P}_K=\{\mathcal{S}_c\subseteq \widehat{\boldsymbol{\mathcal{Z}}}:|\mathcal{S}_c|=K\}$;

    \FOR{each candidate subset $\mathcal{S}_c\in \mathcal{P}_K$}
        \STATE $\widehat{T}_{\mathcal{S}_c}\gets 0$;
        \FOR{each unordered pair $\{\widehat{\mathcal{Z}}_i,\widehat{\mathcal{Z}}_j\}\subseteq \mathcal{S}_c$ with $i<j$}
            \STATE $\widehat{T}_{\mathcal{S}_c}
            \gets
            \widehat{T}_{\mathcal{S}_c}+\widehat{\mathbf{M}}_{ij}$;
        \ENDFOR
    \ENDFOR

    \STATE
    $\widehat{\mathcal{S}}
    \gets
    \arg\min_{\mathcal{S}_c\in\mathcal{P}_K}
    \widehat{T}_{\mathcal{S}_c}$;

    \RETURN $\widehat{\mathcal{S}}$ and $\widehat{T}_{\mathcal{S}}$. 
\end{algorithmic}
\end{algorithm}

Based on the above three components, the entire procedure is summarized in Algorithm~\ref{algorithm-CAT}. 
Overall, the algorithm proceeds in three main stages. 
First, it adjusts for covariates by residualizing the treatment, outcome, and candidate IVs with respect to $\mathbf W$ (Lines~2--6). 
Second, it estimates the auxiliary variables for each candidate IV by estimating the corresponding function $h_i(\cdot)$, as described in $\mathcal{Q}2$ (Lines~7--14). 
Third, it searches over candidate IV subsets of size $K$ and evaluates each subset using the CAT-based score defined in $\mathcal{Q}3$ (Lines~15--26). 
The subset with the smallest empirical CAT score is returned as the estimated IV set that is most consistent with the CAT condition.

Below, we establish the correctness of the CAT algorithm in selecting a valid IV set as the sample size tends to infinity.

\begin{Theorem}[\textbf{Correctness of the CAT Algorithm}]
\label{The-algorithm-correctness}
Assume that the observed data $\{X,Y,\mathbf{W},\mathbf{Z}\}$ are generated from CAM-NCE, and that the candidate IV set $\mathbf{Z}$ contains at least $K$ valid IVs. 
Suppose that Assumptions~\ref{Ass-completeness}--\ref{Ass-higher-order-condition} hold, and that the estimators used in Algorithm~\ref{algorithm-CAT} for covariate adjustment, auxiliary-variable construction, and distance correlation are consistent. 
Then, as the sample size tends to infinity, Algorithm~\ref{algorithm-CAT} outputs a valid IV subset $\widehat{\mathcal{S}}$ with $|\widehat{\mathcal{S}}|=K$. 
In particular, if $\mathbf{Z}$ contains exactly $K$ valid IVs, then Algorithm~\ref{algorithm-CAT} outputs the full valid IV set.
\end{Theorem}

\begin{proof}
See Appendix~\ref{proof-The-algorithm-correctness} for its proof. 
\end{proof}

\noindent\textbf{Computational complexity.} 
We finally analyze the computational complexity of the CAT algorithm. 
Let $n$ denote the sample size, $m=|\mathbf{Z}|$ denote the number of candidate IVs, and $p=|\mathbf{W}|$ denote the number of covariates. 
The time complexity consists of four main components:
\begin{itemize}[leftmargin=15pt,itemsep=0pt,topsep=0pt,parsep=0pt]
    \item[1.] \textbf{Covariate residualization.} 
    Regressing $X$, $Y$, and the $m$ candidate IVs on $\mathbf{W}$ to obtain residualized variables costs $\mathcal{O}(nmp^2)$. 

\item[2.] \textbf{Auxiliary-variable estimation:} 
For each candidate IV, we estimate $h_i(\cdot)$ using the estimator adopted in this work, i.e., the semiparametric control-function estimator for non-constant effects and two-stage least squares for constant effects. 
This step costs $\mathcal{O}(mn)$ in total.

    \item[3.] \textbf{Pairwise distance-correlation calculation.} 
    Computing the directed pairwise distance correlations for all candidate IV pairs costs $\mathcal{O}(m^2n^2)$ using the standard distance-correlation estimator.

    \item[4.] \textbf{Candidate subset selection.} 
    After the pairwise distance-correlation matrix is computed, evaluating all candidate subsets of size $K$ costs $\mathcal{O}\!\left(\binom{m}{K}K^2\right)$.
\end{itemize}
Hence, the overall computational complexity is
\[
\mathcal{O}\!\left(nmp^2 + mn + m^2n^2 + \binom{m}{K}K^2\right).
\]

\section{Experiments}\label{Sec-experimets}
In this section, we evaluate the proposed CAT method on synthetic datasets generated under both CAM-CE and CAM-NCE, corresponding to the constant-effect and non-constant-effect settings, respectively. Our goal is to examine whether CAT can correctly distinguish valid IV sets from invalid ones under different types of IV assumption violations. We first consider the constant-effect setting, where representative existing IV methods are applicable and thus serve as baselines, and then proceed to the more general non-constant-effect setting, which is the primary focus of this paper. It is noteworthy that the simulated data are generated solely according to the causal additive model in Equation~\eqref{Eq-Main-Model}; no additional constraints are imposed to ensure that the distributional condition in Assumption~\ref{Ass-higher-order-condition} holds. Our source code is available at \url{https://github.com/guoxichen0/CAT}.

Across both settings, we consider three representative invalid IV scenarios involving both valid and invalid IVs. 
In Case~1, the invalid IVs violate the exclusion restriction condition~($\mathcal{C}2$); 
in Case~2, the invalid IVs violate the exogeneity condition~($\mathcal{C}3$); 
and in Case~3, the invalid IVs violate both the exclusion restriction~($\mathcal{C}2$) and exogeneity~($\mathcal{C}3$) conditions. 
Each experiment is repeated 100 times with independently generated data. 
All noise terms are independently drawn from \(U(-1,1)\). 
For each case, we vary the sample size over \(n\in\{1000,3000,5000\}\).

\subsection{Synthetic Data under Constant Effects}\label{sec_synthetic_constant}

We first evaluate CAT under the CAM-CE framework, where the structural response function takes the linear form $f(X)=\beta X$, yielding a constant causal effect. Following~\cite{guo2018confidence} and related works, the true constant causal effect is fixed at \(\beta=1\). Other constant coefficients in the structural equations are independently sampled from \([-1.5,-0.5]\cup[0.5,1.5]\). We report comparison results under nonlinear structural components with constant treatment effects, settings with covariates, and the ALICE model.

Since existing IV selection methods are primarily developed for constant-effect or linear models, the CAM-CE setting enables direct comparison with the following representative baselines: 
\begin{enumerate}[leftmargin=15pt,itemsep=0pt,topsep=0pt,parsep=0pt]
    \item NAIVE: the least-squares regression coefficient of \(Y\) on \(X\);
    \item MR-Egger~\citep{bowden2015mendelian}: implemented using the code available at \url{https://academic.oup.com/ije/article/44/2/512/754653/};
    \item TSHT~\citep{guo2018confidence}: implemented using the code available at \url{https://cran.r-project.org/web/packages/RobustIV/};
    \item CIIV~\citep{windmeijer2021confidence}: implemented using the code available at \url{https://github.com/xlbristol/CIIV/};
    \item sisVIVE~\citep{kang2016instrumental}: implemented using the code available at \url{https://cran.r-project.org/web/packages/sisVIVE/};
    \item IV-tetrad~\citep{silva2017learning}: implemented using the code available at \url{https://www.homepages.ucl.ac.uk/~ucgtrbd/code/iv_discovery/}.
\end{enumerate}
\noindent\textbf{Metrics.} We evaluate the methods through their resulting causal effect estimates, summarized using boxplots, where estimates closer to the true effect with smaller variability and dispersion indicate better performance. 
For methods that first select IVs, including TSHT, CIIV, sisVIVE, IV-tetrad, and CAT, we apply the same IV estimator as in~\cite{guo2018confidence} after IV selection to ensure a fair comparison. NAIVE and MR-Egger directly produce causal effect estimates and are therefore evaluated using their original outputs. 

\begin{figure*}[h]
\centering
\subfloat[Case 1]{%
    \includegraphics[width=0.33\textwidth]{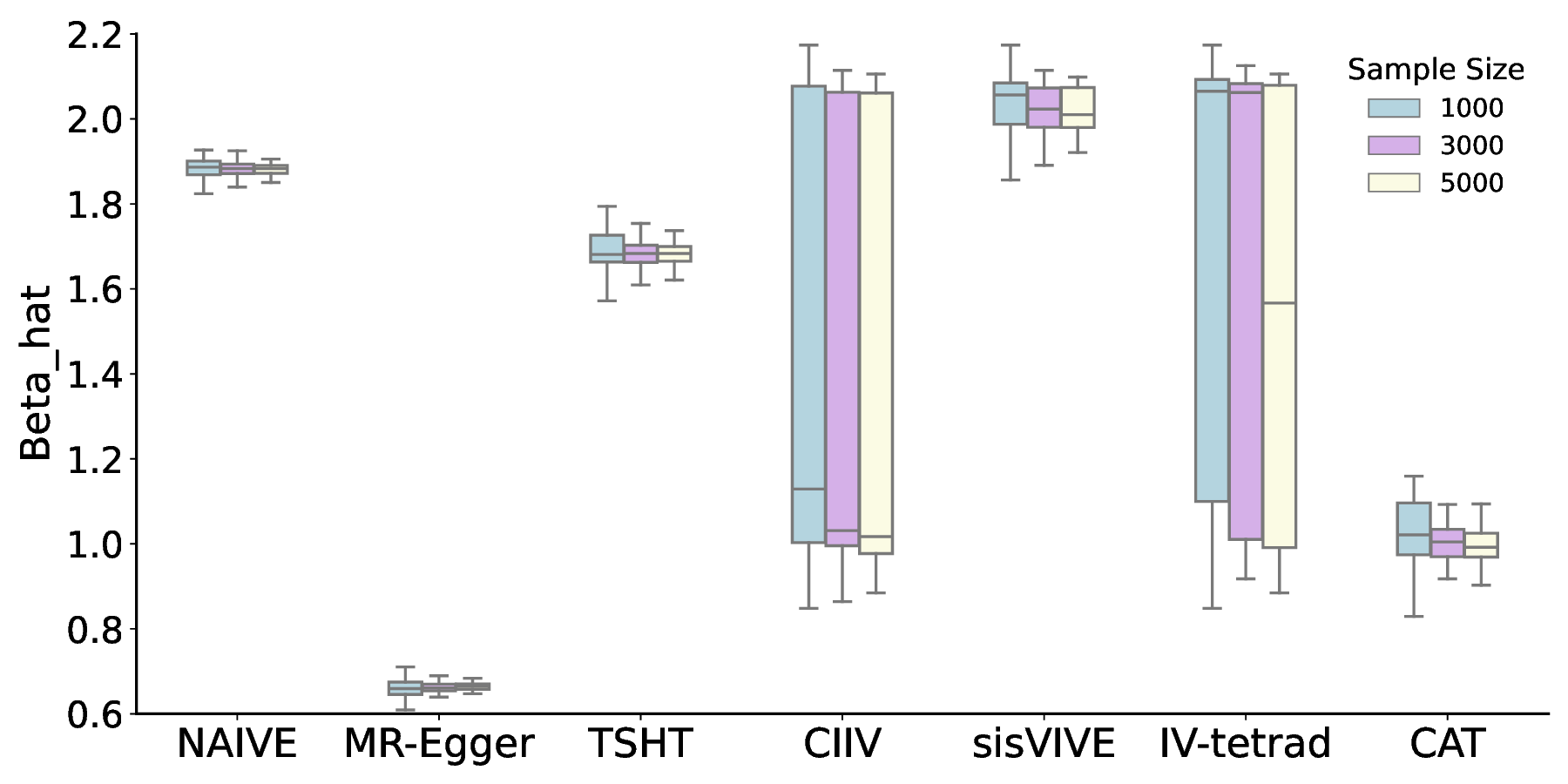}%
    \label{fig:Case1_nonlinear}}
\hfil
\subfloat[Case 2]{%
    \includegraphics[width=0.33\textwidth]{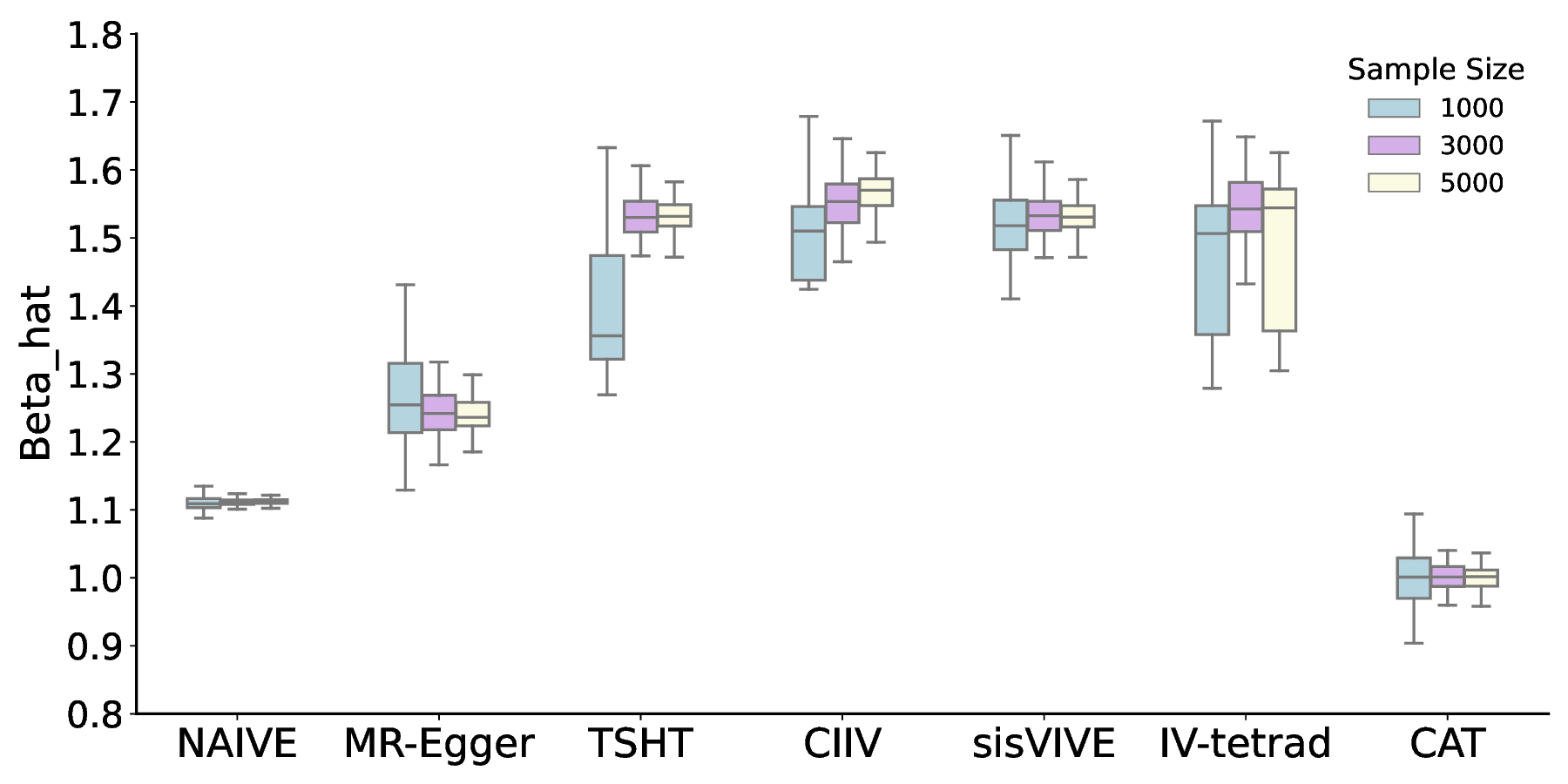}%
    \label{fig:Case2_nonlinear}}
\hfil
\subfloat[Case 3]{%
    \includegraphics[width=0.33\textwidth]{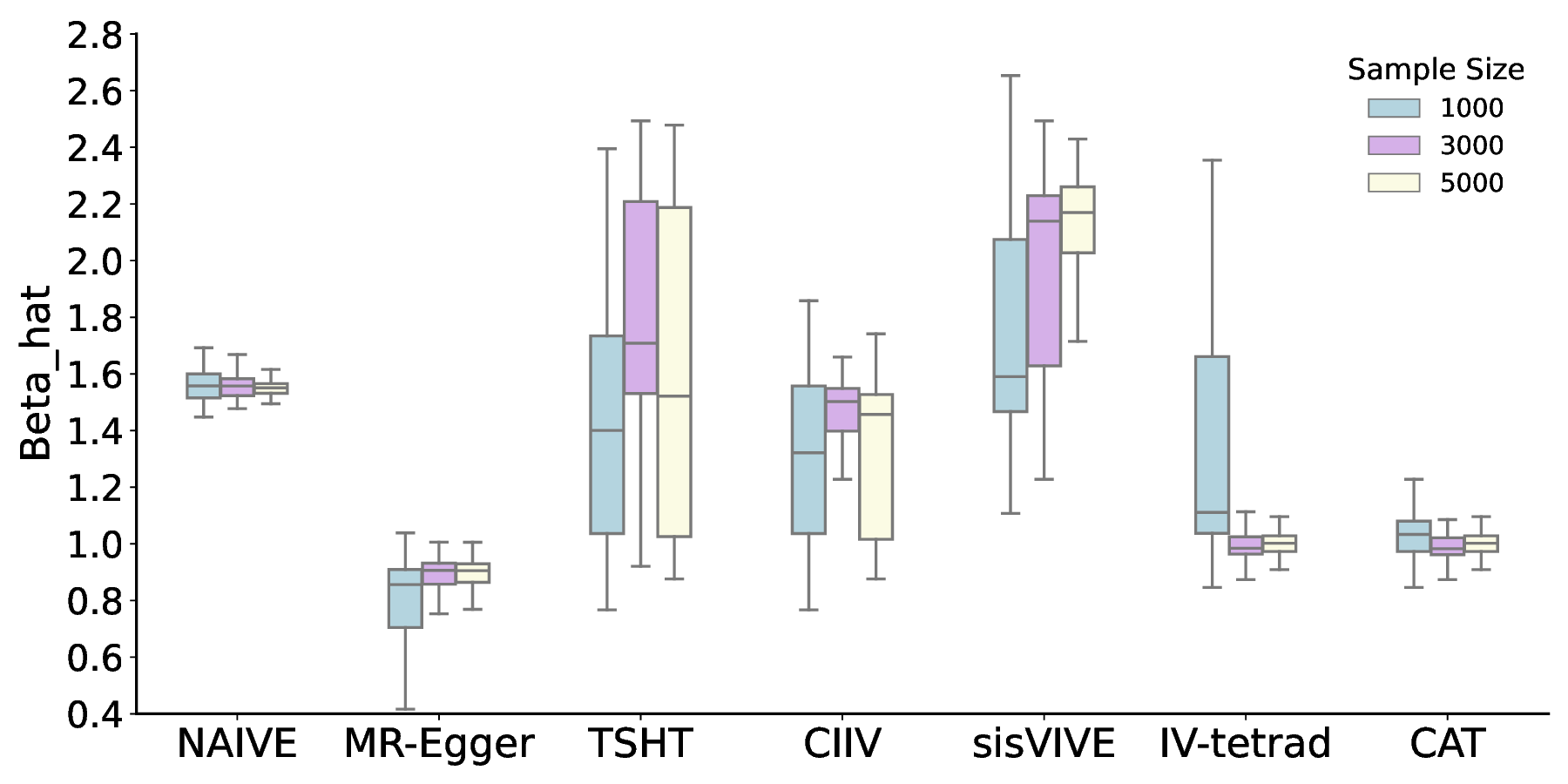}%
    \label{Case3_nonlinear}}
\caption{Performance of NAIVE, MR-Egger, TSHT, CIIV, sisVIVE, IV-tetrad, and CAT across three different cases in the CAM-CE framework.}
\label{fig:compare-method-constant-nonlinear-model}
\end{figure*}

\begin{figure*}[h]
\centering
\subfloat[Case 1]{%
    \includegraphics[width=0.33\textwidth]{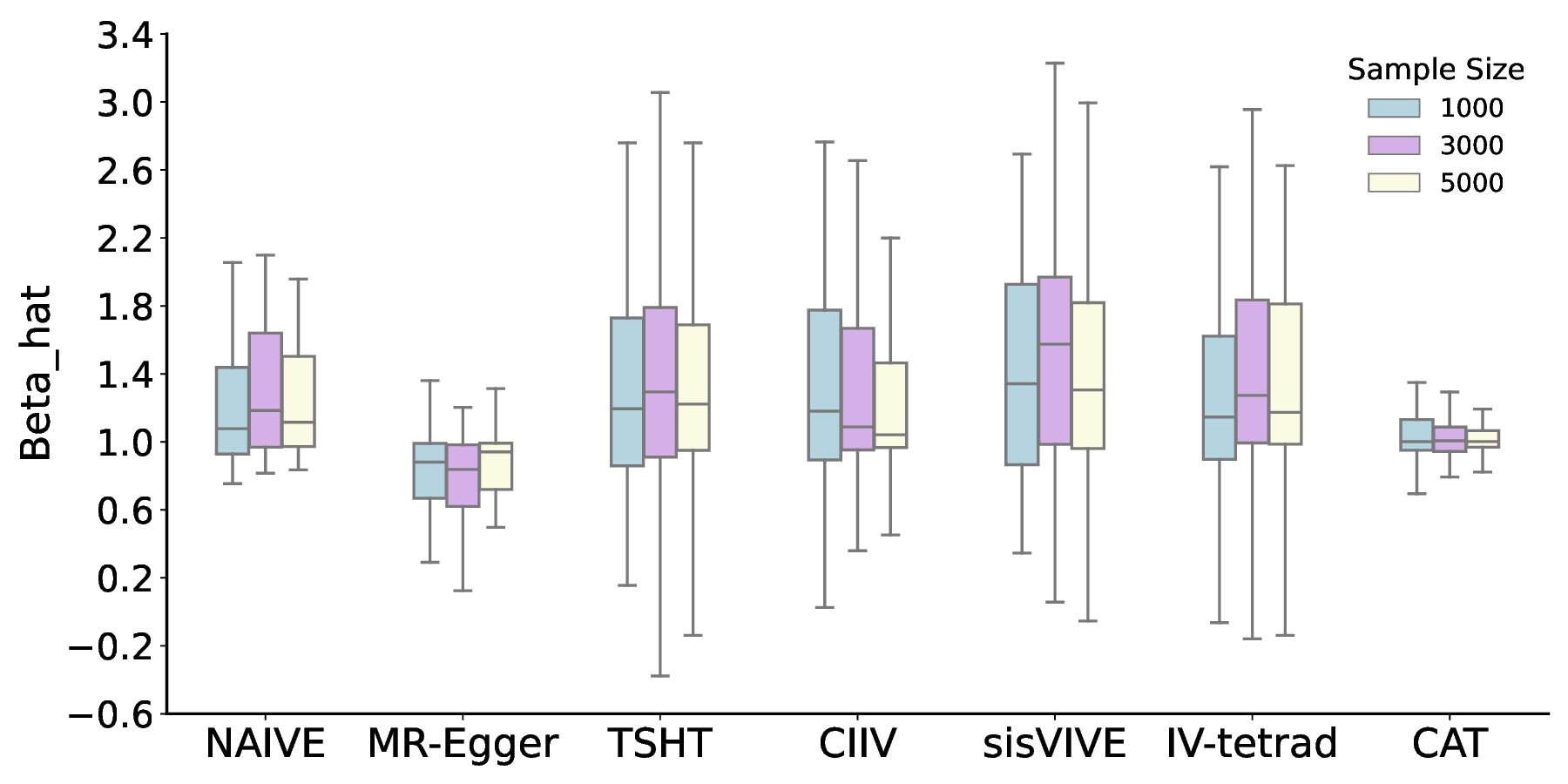}%
    \label{fig:Case1_covariates}}
\hfil
\subfloat[Case 2]{%
    \includegraphics[width=0.33\textwidth]{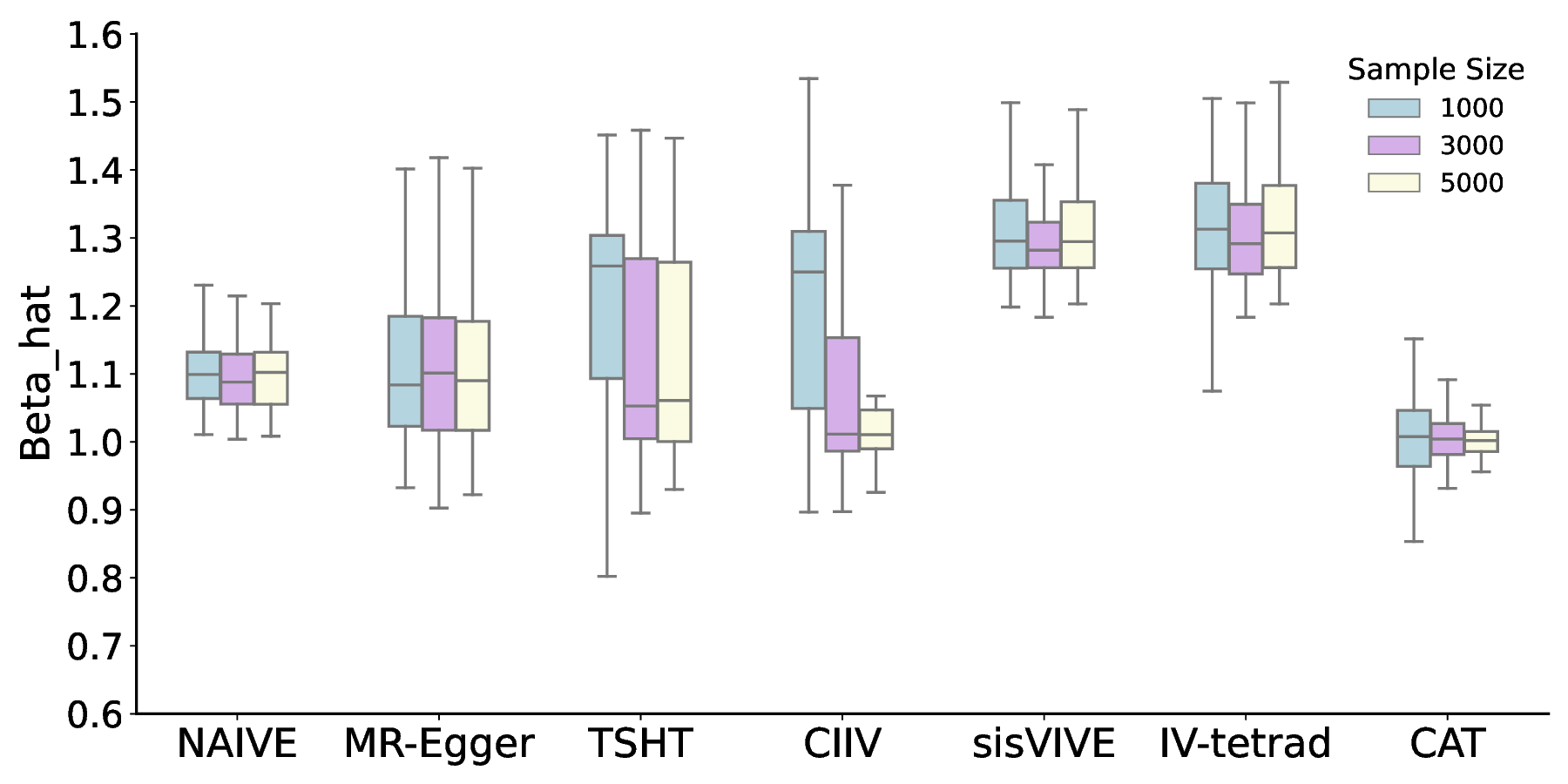}%
    \label{fig:Case2_covariates}}
\hfil
\subfloat[Case 3]{%
    \includegraphics[width=0.33\textwidth]{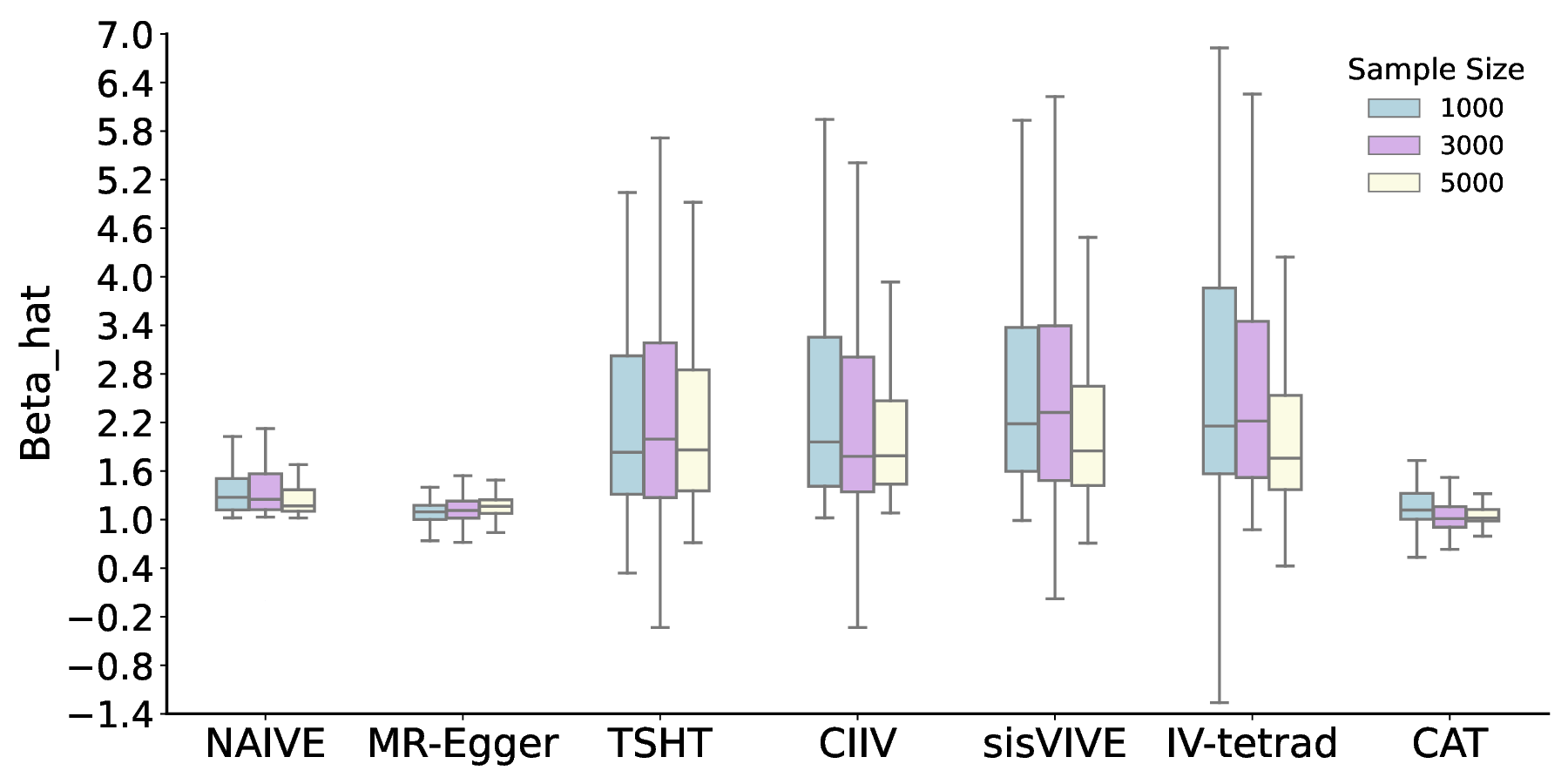}%
    \label{Case3_covariates}}

\caption{Performance of NAIVE, MR-Egger, TSHT, CIIV, sisVIVE, IV-tetrad, and CAT across three different cases with covariates $\mathbf{W}$ in the CAM-CE framework.}
\label{fig:compare-method-constant-covariates-model}
\end{figure*}

\begin{figure*}[h]
\centering
\subfloat[Case 1]{%
    \includegraphics[width=0.32\textwidth]{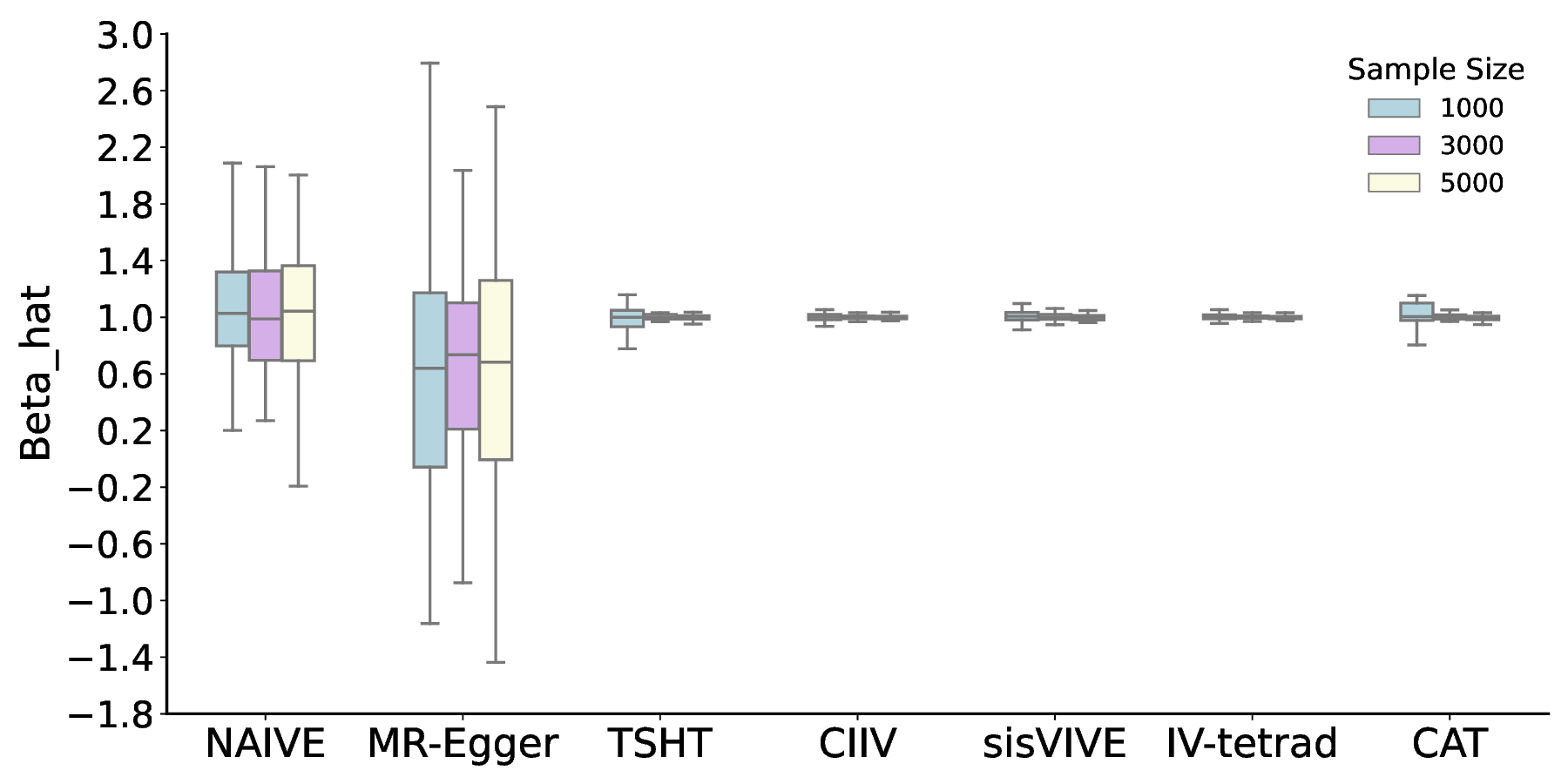}%
    \label{fig:Case1_linear}}
\hfil
\subfloat[Case 2]{%
    \includegraphics[width=0.32\textwidth]{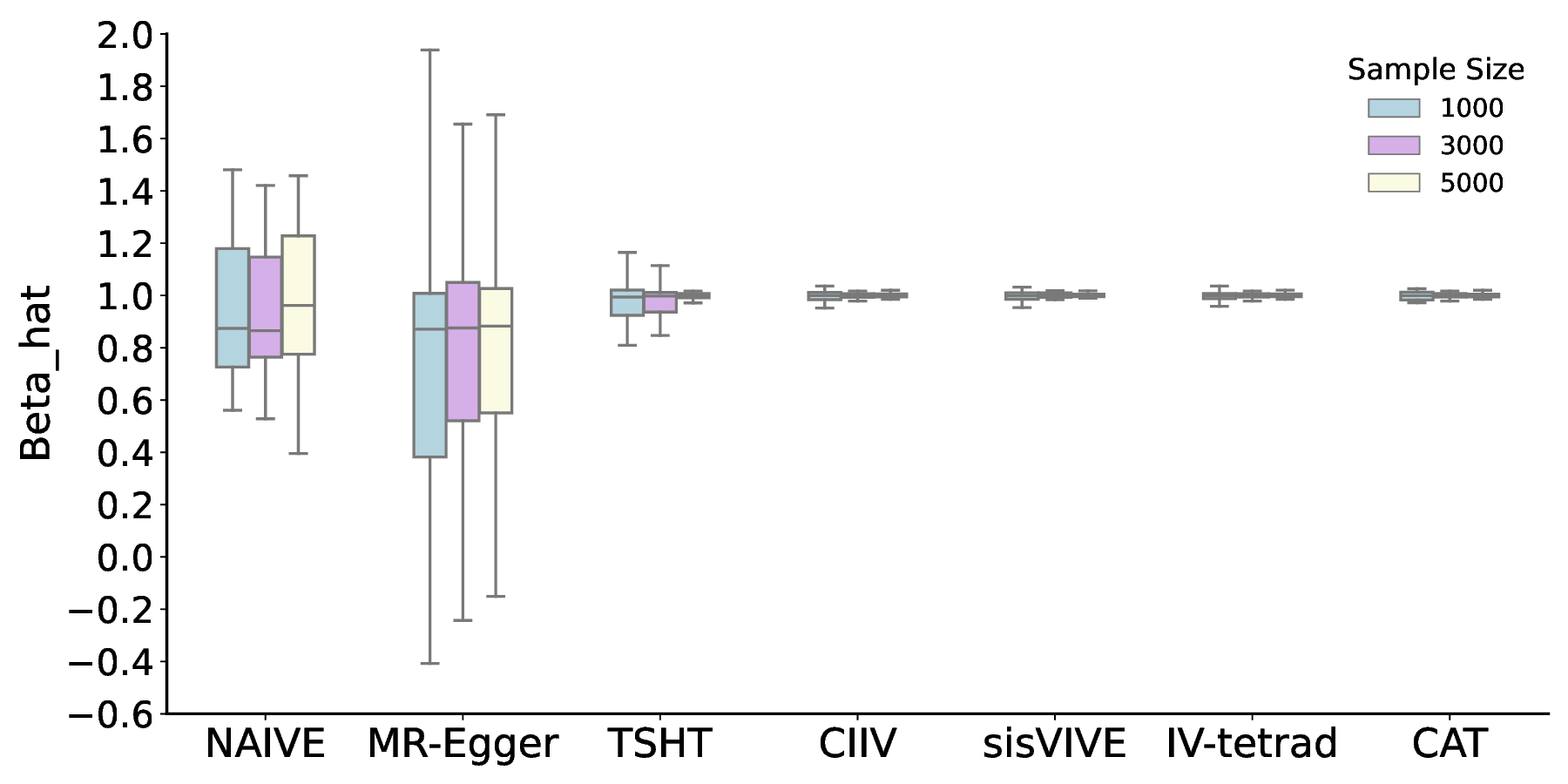}%
    \label{fig:Case2_linear_uniform}}
\hfil
\subfloat[Case 3]{%
    \includegraphics[width=0.32\textwidth]{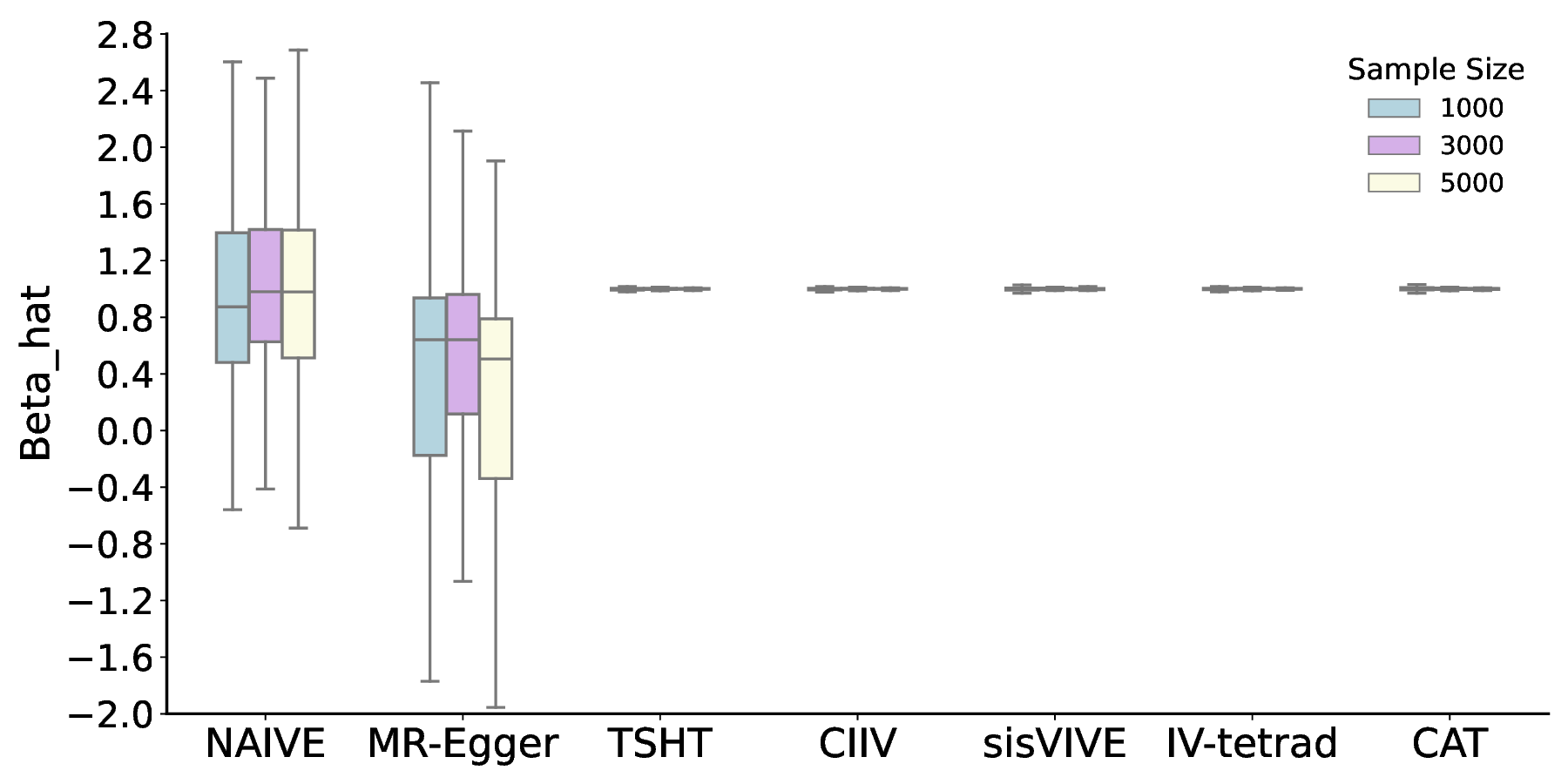}%
    \label{Case3_linear_uniform}}

\caption{Performance of NAIVE, MR-Egger, TSHT, CIIV, sisVIVE, IV-tetrad, and CAT across three different cases in the ALICE model.}
\label{fig:compare-method-constant-linear-model}
\end{figure*}

\noindent\textbf{Results.} 
The results are presented in Figures~
\ref{fig:compare-method-constant-nonlinear-model}--\ref{fig:compare-method-constant-linear-model}. Figures~\ref{fig:compare-method-constant-nonlinear-model} and \ref{fig:compare-method-constant-covariates-model} report the results under nonlinear structural components with constant treatment effects, without and with covariates, respectively. As expected, the proposed CAT algorithm consistently outperforms the other methods across all three cases and sample sizes, exhibiting relatively small variance and producing estimates closest to the true causal effect. In contrast, the NAIVE method performs poorly in all cases due to unmeasured confounders $\mathbf{U}$. We observed that the other comparison methods perform poorly across all cases because they rely on the assumption of linearity, whereas the data generation process is partially nonlinear. Additionally, we found that the MR-Egger algorithm yields inaccurate results. A possible reason for this is that, in addition to the linearity assumption, this method requires the InSIDE assumption, which states that the instruments' pleiotropic effects on the outcome $Y$ are uncorrelated with their effects on the treatment $X$. 
Furthermore, Figure~\ref{fig:compare-method-constant-linear-model} reports the results under the ALICE model. In this linear constant-effect setting, CAT performs well across all three cases and yields estimates close to the true causal effect. Its performance is comparable to TSHT, CIIV, sisVIVE, and IV-tetrad, which is expected because these methods are designed for linear constant-effect settings. In contrast, NAIVE and MR-Egger exhibit larger variability or bias across all cases.

\subsection{Synthetic Data under Non-Constant Effects}\label{sec_synthetic_non_constant}

We next evaluate CAT under the CAM-NCE framework, where $f(\cdot)$ is nonlinear, and the causal effect varies with the treatment level.
Existing IV-selection methods considered above rely on linear or constant-effect specifications and are not designed for the CAM-NCE setting. Applying them here would therefore evaluate the methods under model misspecification rather than provide a meaningful comparison of IV-set identification. To the best of our knowledge, no existing method provides a directly comparable IV-set identification procedure under the CAM-NCE framework considered here. We thus focus on systematically evaluating CAT across a range of non-constant-effect settings.
Specifically, we assess the performance of CAT from three perspectives: 
(i) varying the causal effect function from \(X\) to \(Y\), including \(\log\), \(\sin\), \(\cos\), \(\mathrm{quadratic\ polynomial}\), \(\mathrm{cubic\ polynomial}\), \(\log(\mathrm{quad})\), and \(\exp(\mathrm{quad})\), with \(K=2\) valid IVs among five candidate IVs; 
(ii) varying the number of valid IVs, \(K\in\{2,3,5\}\), among ten candidate IVs; and 
(iii) varying the number of covariates, \(|\mathbf W|\in\{2,3,5\}\), with \(K=2\) valid IVs among five candidate IVs.

\noindent\textbf{Metrics.} We evaluate IV-set identification performance using the Error Rate (ER). 
A trial is regarded as successful only when the IV set identified by the method exactly matches the true valid IV set. Accordingly, ER is defined as the proportion of trials for which the identified IV set differs from the true valid IV set, with a lower ER indicating better identification performance. 

\begin{table}[!h]
\centering
\small
\caption{Performance of IV-validity Testing under Different Causal Effect Functions in the CAM-NCE Framework}
\setlength{\tabcolsep}{5pt} 
\begin{tabular}{ccccc}
\toprule
 & \multicolumn{1}{c}{Sample sizes} & Size=1000 & Size=3000 & Size=5000 \\
\cmidrule(lr){2-2} \cmidrule(lr){3-5}
Cases & Functions $f(X)$ 
& \textbf{ER}$\boldsymbol{\downarrow}$ & \textbf{ER}$\boldsymbol{\downarrow}$ & \textbf{ER}$\boldsymbol{\downarrow}$ \\
\midrule

% -------- Case 1 -------- %
\multirow{7}{*}{Case 1}
& Log                  & 0.00 & 0.00 & 0.00 \\
& Sin                  & 0.00 & 0.00 & 0.00 \\
& Cos                  & 0.01 & 0.00 & 0.00 \\
& Quadratic Poly       & 0.00 & 0.00 & 0.00 \\
& Cubic Poly           & 0.00 & 0.00 & 0.00 \\
& Log(quad)            & 0.00 & 0.00 & 0.00 \\
& Exp(quad)            & 0.00 & 0.00 & 0.00 \\
\midrule

% -------- Case 2 -------- %
\multirow{7}{*}{Case 2}
& Log                  & 0.00 & 0.00 & 0.00 \\
& Sin                  & 0.00 & 0.00 & 0.00 \\
& Cos                  & 0.00 & 0.00 & 0.00 \\
& Quadratic Poly       & 0.00 & 0.00 & 0.00 \\
& Cubic Poly           & 0.00 & 0.00 & 0.00 \\
& Log(quad)            & 0.00 & 0.00 & 0.00 \\
& Exp(quad)            & 0.00 & 0.00 & 0.00 \\
\midrule

% -------- Case 3 -------- %
\multirow{7}{*}{Case 3}
& Log                  & 0.00 & 0.00 & 0.00 \\
& Sin                  & 0.00 & 0.00 & 0.00 \\
& Cos                  & 0.00 & 0.00 & 0.00 \\
& Quadratic Poly       & 0.00 & 0.00 & 0.00 \\
& Cubic Poly           & 0.00 & 0.00 & 0.00 \\
& Log(quad)            & 0.04 & 0.00 & 0.00 \\
& Exp(quad)            & 0.00 & 0.00 & 0.00 \\
\bottomrule
\end{tabular}

\vspace{2pt}
\begin{tablenotes}
\footnotesize
\item Note: ``Quadratic Poly/quad'' and ``Cubic Poly'' denote quadratic and cubic polynomial functions, respectively. $\downarrow$ indicates that lower values are better.

\end{tablenotes}

\label{Table-nonconstant}
\end{table}

\begin{table}[!h]
\centering
\small
\caption{Performance of IV-validity Testing for Different Numbers \(K\) of Valid IVs in the CAM-NCE Framework}
\begin{tabular}{ccccc}
\toprule
& \multicolumn{1}{c}{Sample sizes} & Size=1000 & Size=3000 & Size=5000 \\
\cmidrule(lr){2-2} \cmidrule(lr){3-5}
$K$ & Cases 
& \textbf{ER}$\boldsymbol{\downarrow}$ 
& \textbf{ER}$\boldsymbol{\downarrow}$ 
& \textbf{ER}$\boldsymbol{\downarrow}$ \\
\midrule

% -------- |W| = 2 -------- %
\multirow{3}{*}{2}
& Case 1 & 0.04  & 0.00  & 0.00 \\
& Case 2 & 0.00  & 0.00  & 0.00 \\
& Case 3 & 0.01  & 0.01  & 0.00 \\
\midrule

% -------- |W| = 3 -------- %
\multirow{3}{*}{3}
& Case 1 & 0.01  & 0.00  & 0.00 \\
& Case 2 & 0.00  & 0.00  & 0.00 \\
& Case 3 & 0.00  & 0.00  & 0.00 \\
\midrule

% -------- |W| = 5 -------- %
\multirow{3}{*}{5}
& Case 1 & 0.00  & 0.00  & 0.00 \\
& Case 2 & 0.00  & 0.00  & 0.00 \\
& Case 3 & 0.00  & 0.00  & 0.00 \\

\bottomrule
\end{tabular}

\vspace{2pt}
\begin{tablenotes}
\footnotesize
\item Note: $K$ denotes the number of valid IVs to consider. $\downarrow$ indicates that lower values are better.

\end{tablenotes}

\label{Table-nonconstant-K}
\end{table}

\begin{table}[h]
\centering
\small
\caption{Performance of IV-validity Testing with Covariates in the CAM-NCE Framework}
\begin{tabular}{ccccc}
\toprule
 & \multicolumn{1}{c}{Sample sizes} & Size=1000 & Size=3000 & Size=5000 \\
\cmidrule(lr){2-2} \cmidrule(lr){3-5}
$|\mathbf{W}|$ & Cases 
& \textbf{ER}$\boldsymbol{\downarrow}$ 
& \textbf{ER}$\boldsymbol{\downarrow}$ 
& \textbf{ER}$\boldsymbol{\downarrow}$ \\
\midrule

% -------- |W| = 2 -------- %
\multirow{3}{*}{2}
& Case 1 & 0.00 & 0.00 & 0.00 \\
& Case 2 & 0.08 & 0.07 & 0.05 \\
& Case 3 & 0.04 & 0.04 & 0.01 \\
\midrule

% -------- |W| = 3 -------- %
\multirow{3}{*}{3}
& Case 1 & 0.02 & 0.00 & 0.00 \\
& Case 2 & 0.10 & 0.09 & 0.08 \\
& Case 3 & 0.14 & 0.10 & 0.04 \\
\midrule

% -------- |W| = 5 -------- %
\multirow{3}{*}{5}
& Case 1 & 0.00 & 0.00 & 0.00 \\
& Case 2 & 0.14 & 0.10 & 0.08 \\
& Case 3 & 0.15 & 0.12 & 0.10 \\

\bottomrule
\end{tabular}

\vspace{2pt}
\begin{tablenotes}
\footnotesize
\item Note: $|\mathbf{W}|$ denotes the number of covariates. $\downarrow$ indicates that lower values are better.

\end{tablenotes}

\label{Table-nonconstant-covariates}
\end{table}

\noindent\textbf{Results.} The results are summarized in Tables~\ref{Table-nonconstant}--\ref{Table-nonconstant-covariates}.
Overall, the proposed method achieves consistently low error rates in all considered settings, suggesting that it can reliably identify the valid IV set under non-constant causal effects. 
Table~\ref{Table-nonconstant} reports the results for different forms of the causal effect function from \(X\) to \(Y\). 
Across logarithmic, trigonometric, polynomial, and exponential functions, the ER is close to zero for all sample sizes, indicating that the proposed method is insensitive to the specific functional form of the non-constant causal effect. 
Table~\ref{Table-nonconstant-K} further examines the influence of the number of valid IVs in the candidate IV set. 
The ER remains low for different choices of \(K\) and generally decreases as the sample size increases, showing that the method is stable across different proportions of valid IVs. 
This is in contrast to methods relying on the \emph{Majority Rule} or \emph{Plurality Rule} assumptions (see Section~\ref{Sec-introduction}), which require restrictions on the proportion of valid IVs. 
Finally, Table~\ref{Table-nonconstant-covariates} presents the results when covariates \(\mathbf W\) are included. 
Although the ER slightly increases as the number of covariates grows, it decreases with larger sample sizes. This suggests that the proposed method remains effective in the presence of covariates, while benefiting from increased sample size.

\section{Applications to Real-world Data}\label{Section-real-world}

In this section, we apply CAT to three real-world datasets spanning sociology, economics, and behavioral science to examine its practical utility. Since the ground-truth validity of candidate IVs is not directly observable in real-world applications, we use the IV specifications proposed in prior studies as reference benchmarks and compare the conclusions obtained by CAT with those reported in the literature. We report both the CAT-based distance-correlation scores and the corresponding p-values from distance-correlation independence tests.
% , where $H_0: A \CI B$ is tested against $H_1: A \nCI B$. 
Here, the significance level is set to $\alpha=10/n$, where $n$ denotes the sample size used in the test.

\subsection{{Colonial Origins Data \citep{acemoglu2001colonial}}}
This dataset examines the impact of social systems on economic development. 
After excluding observations with missing values, it contains five key variables across 63 countries: \emph{Mortality} $(M_{or})$, \emph{Euro1990} $(E_{uro})$, \emph{Latitude} $(L_{at})$, \emph{Institutions} $(I_{ns})$, and \emph{Economic Development} $(E_d)$. 
We take the IV specification in~\cite{acemoglu2001colonial} as a reference benchmark and evaluate it under the constant-effect setting. 
The hypothesized model proposed by~\cite{acemoglu2001colonial} is illustrated in Figure~\ref{Fig:real-data-example2}, and the hypothesized data generation mechanism is described as follows: 
\begin{equation}\nonumber
    \begin{aligned}
        I_{ns} &= \alpha_0 + \alpha_1 M_{or} + \alpha_2 E_{uro}+ \alpha_3 L_{at} + \delta,\\
        E_d &= \beta_0 + \beta_1 I_{ns} + \beta_2 L_{at} + \epsilon,
    \end{aligned}
\end{equation}
where $\delta$ and $\epsilon$ are dependent. 
Specifically, we assess the candidate IV set $\{M_{or},E_{uro}\}$ for the causal relation $I_{ns}\to E_d$, conditioning on the covariate $L_{at}$.
\begin{figure}[h]
\centering
\begin{tikzpicture}[
    node distance=1.0cm and 1.4cm,
    line width=0.8pt,
    >=latex,
    font=\footnotesize,
    ivblue/.style={
        rectangle,
        rounded corners=2pt,
        draw=black!65,
        fill=blue!8,
        text width=1.05cm,
        align=center,
        minimum height=0.5cm
    },
    ivyellow/.style={
        rectangle,
        rounded corners=2pt,
        draw=black!65,
        fill=yellow!18,
        text width=1.05cm,
        align=center,
        minimum height=0.5cm
    },
    mainnode/.style={
        rectangle,
        rounded corners=2pt,
        draw=black!65,
        fill=white,
        text width=1.05cm,
        align=center,
        minimum height=0.5cm
    },
    uinode/.style={
        rectangle,
        rounded corners=2pt,
        draw=black!65,
        fill=gray!10,
        text width=2cm,
        align=center,
        minimum height=1cm
    },
    solidarrow/.style={draw=black, ->}
]

% ---------------- Nodes ----------------

% Latitude
\node[mainnode] (Lat) at (-3.0,0) {$L_{at}$};

% IVs
\node[ivyellow] (Euro) at (-1.25,0.8) {$E_{uro}$};
\node[ivblue]   (Mor)  at (-1.25,1.7) {$M_{or}$};

% Treatment and outcome
\node[mainnode] (Ins) at (0.8,0) {$I_{ns}$};
\node[mainnode] (Ed)  at (3.0,0) {$E_d$};

% Unmeasured confounder
\node[uinode] (U) at (1.9,1.55)
{\emph{Cultural}\\\emph{difference}};

% ---------------- Edges ----------------

% Latitude -> IVs
\draw[solidarrow] (Lat) -- (Euro);
\draw[solidarrow] (Lat) -- (Mor);

% IVs -> institutions
\draw[solidarrow] (Euro) -- (Ins);
\draw[solidarrow] (Mor) -- (Ins);

% Latitude -> institutions
\draw[solidarrow] (Lat) -- (Ins);

% Institutions -> development
\draw[solidarrow] (Ins) -- (Ed);

% Unmeasured confounder
\draw[solidarrow] (U) -- (Ins);
\draw[solidarrow] (U) -- (Ed);

% Latitude -> development
\draw[solidarrow, bend right=13]
    (Lat.south) to (Ed.south);

\end{tikzpicture}

\caption{Graphical illustration of an IV model for estimating the causal effect
of institutions ($I_{ns}$) on economic development ($E_d$)
\citep{acemoglu2001colonial}.}
\label{Fig:real-data-example2}
\end{figure}
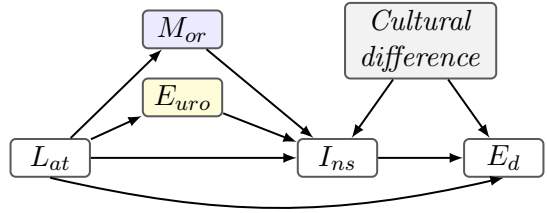

\noindent\textbf{Results.}
Since there are only two candidate IVs, we apply CAT directly to this pair and obtain a CAT-based distance-correlation score of $0.50$. 
We then conduct distance-correlation independence tests for the two directed CAT relations, namely between $\mathcal{V}_{M_{or}}$ and $E_{uro}$, and between $\mathcal{V}_{E_{uro}}$ and $M_{or}$, obtaining $p$-values of $0.21$ and $0.25$, respectively. 
Therefore, we do not reject the corresponding CAT independence relations, providing no evidence against $\{M_{or},E_{uro}\}$ as a valid IV set for $I_{ns}\to E_d$. 
This result is consistent with the findings of~\cite{acemoglu2001colonial}.

\subsection{\textbf{Children and Mothers' Labor Supply Data~\citep{angrist1996children}}}
This dataset comes from an empirical study on the effect of childbearing on mothers' labor supply. 
After applying the filtering criteria, it contains 254,652 observations. 
We use \emph{more than two children} $(\textit{morekids})$ as the treatment and \emph{weeks worked} $(\textit{weeksm1})$ as the outcome. 
The candidate IVs include \emph{two boys} $(\textit{boys2})$, \emph{two girls} $(\textit{girls2})$, \textit{AGEQK}, \textit{AGEQ2ND}, \textit{KIDCOUNT}, \textit{YOBM}, \textit{nonmomil}, \textit{educm}, \textit{hsormore}, \textit{nonmomi}, \textit{ageqm}, and \textit{agefstd}. 
The covariates $\mathbf W$ include \emph{mother's age at first birth} $(\textit{agem1})$, \emph{father's age at first birth} $(\textit{agefstm})$, \emph{whether the first child is a boy} $(\textit{boy1st})$, \emph{whether the second child is a boy} $(\textit{boy2nd})$, \emph{black mother indicator} $(\textit{blackm})$, \emph{Hispanic mother indicator} $(\textit{hispm})$, and \emph{other-race mother indicator} $(\textit{othracem})$.
We take the IV specification in~\cite{angrist1996children} as a reference benchmark and evaluate it under the constant-effect setting. 
% In particular, Angrist and Evans~\cite{angrist1996children} used the sex composition of the first two children, including \textit{boys2} and \textit{girls2}, as instruments for childbearing. 
Due to the large sample size and the quadratic computational cost of distance correlation, we randomly subsample $5\%$ of the data and average the results over 10 repeated tests. The valid IVs hypothesized model proposed by~\cite{angrist1996children} is illustrated in Figure~\ref{Fig:real-data-example}, and the valid IVs hypothesized data generation mechanism is described as follows: 
\begin{equation}\nonumber
    \begin{aligned}
        morekids &= \gamma_0 boys2 + \gamma_1 girls2 + \boldsymbol{\gamma_2}^\top \mathbf{W} + \delta,\\
        weeksm1 &= \beta \cdot morekids + \boldsymbol{\beta_1}^\top (\mathbf{W}\setminus \{boy2nd\}) + \epsilon, 
    \end{aligned}
\end{equation}
where $\delta$ and $\epsilon$ are dependent, $\mathbf{W}\setminus \{boy2nd\}$ represents the set of all elements in covariates $\mathbf{W}$ after removing the variable $boy2nd$. 

\begin{figure}[h]
\centering
\begin{tikzpicture}[
    node distance=1.0cm and 1.4cm,
    line width=0.8pt,
    >=latex,
    font=\footnotesize,
    ivblue/.style={
        rectangle,
        rounded corners=2pt,
        draw=black!65,
        fill=blue!8,
        text width=1.3cm,
        align=center,
        minimum height=0.6cm
    },
    ivyellow/.style={
        rectangle,
        rounded corners=2pt,
        draw=black!65,
        fill=yellow!18,
        text width=1.3cm,
        align=center,
        minimum height=0.6cm
    },
    mainnode/.style={
        rectangle,
        rounded corners=2pt,
        draw=black!65,
        fill=white,
        minimum width=1.9cm,
        align=center,
        minimum height=0.6cm
    },
    uinode/.style={
        rectangle,
        rounded corners=2pt,
        draw=black!65,
        fill=gray!10,
        text width=2.3cm,
        align=center,
        minimum height=1cm
    },
    solidarrow/.style={
        draw=black,
        ->
    }
]

% ---------------- Nodes ----------------

% Covariates
\node[mainnode] (W) at (-3.0,0)
{$\mathbf{W}\setminus \{boy2nd\}$};

% \node[mainnode] (b) at (0.8,1.7)
% {${boy2nd}$};

% IVs
\node[ivyellow] (Boys) at (-1.25,0.8)
{$boys2$};

\node[ivblue] (Girls) at (-1.25,1.7)
{$girls2$};

% Treatment and outcome
\node[mainnode] (Morekids) at (0.8,0)
{$morekids$};

\node[mainnode] (Weeksm) at (3.0,0)
{$weeksm1$};

% Unobserved confounder
\node[uinode] (U) at (1.9,1.55)
{\emph{Unobserved}\\\emph{confounder}};

% ---------------- Edges ----------------

% Covariates -> IVs
\draw[solidarrow] (W) -- (Boys);
\draw[solidarrow] (W) -- (Girls);

% IVs -> treatment
\draw[solidarrow] (Boys) -- (Morekids);
\draw[solidarrow] (Girls) -- (Morekids);

% Covariates -> treatment
\draw[solidarrow] (W) -- (Morekids);

% Treatment -> outcome
\draw[solidarrow] (Morekids) -- (Weeksm);

% Unobserved confounder
\draw[solidarrow] (U) -- (Morekids);
\draw[solidarrow] (U) -- (Weeksm);

% Covariates -> outcome
\draw[solidarrow, bend right=13]
    (W.south) to (Weeksm.south);

\end{tikzpicture}

\caption{Graphical illustration of a valid IV model for estimating the causal
effect of childbearing ($morekids$) on mother's labor supply ($weeksm1$)
\citep{angrist1996children}.}
\label{Fig:real-data-example}
\end{figure}
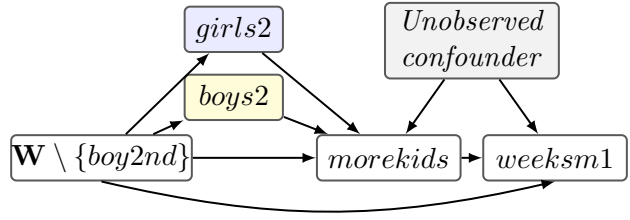

\noindent\textbf{Results.} Using CAT with $K=2$, we find that the candidate set $\{\textit{boys2},\textit{girls2}\}$ achieves the smallest CAT-based distance-correlation score, with $\operatorname{dCor}=0.022$. 
We further conduct distance-correlation independence tests for the two directed CAT relations: between $\mathcal{V}_{\textit{boys2}}$ and $\textit{girls2}$, and between $\mathcal{V}_{\textit{girls2}}$ and $\textit{boys2}$, obtaining $p$-values of $0.32$ and $0.34$, respectively. 
Thus, we do not reject the corresponding CAT independence relations, and CAT selects $\{\textit{boys2},\textit{girls2}\}$ as the estimated valid IV set for $\textit{morekids}\to\textit{weeksm1}$. 
This result is consistent with the conclusion of~\cite{angrist1996children}.

\subsection{\textbf{Conflict and Time Preference Data \citep{voors2012violent}}}
This dataset comes from an empirical study on the effect of violent conflict on individual time preferences. 
In our analysis, we focus on the causal effect of \emph{Violence} on \emph{Patience}.
After removing observations with missing values, the dataset contains 266 observations and 15 variables, including the treatment variable \emph{Violence} $(V_{io})$, the outcome variable \emph{Patience} $(P_{at})$, two candidate IVs, \emph{Distance} $(D_{ist})$ and \emph{Altitude} $(A_{lti})$, and 11 covariates $\mathbf W$. 
The covariates include \emph{literate}, \emph{age}, \emph{sex}, \emph{total land holding per capita}, \emph{land Gini coefficient}, \emph{distance to market}, \emph{conflict over land}, \emph{ethnic homogeneity}, \emph{socioeconomic homogeneity}, \emph{population density}, and \emph{per capita total expenditure}. 
We take the IV specification in~\cite{voors2012violent} as a reference benchmark and evaluate it under the non-constant effect setting. The hypothesized model from~\cite{guo2016control} is illustrated in Figure~\ref{Fig:real-data-conflict-Preference}, and the hypothesized generation mechanism is as follows: 
\begin{equation}\nonumber
\begin{aligned}
V_{io}
&=
\alpha_0+ \alpha_1 D_{ist}+ \alpha_2 A_{lti}+ \alpha_3 D_{ist}^2+ \alpha_4 A_{lti}^2\\
&+ \alpha_5 D_{ist} \cdot A_{lti}
+ \boldsymbol{\alpha}_6^{\top}\mathbf{W}
+ \delta, \\
P_{at}
&=
\beta_0
+ \boldsymbol{\beta}_1^{\top}\mathbf{W}
+ \beta_2 V_{io}
+ \beta_3 V_{io}^2
+ \epsilon,
\end{aligned}
\end{equation}
where $\delta$ and $\epsilon$ are dependent. 

\begin{figure}[htp!]
\centering
\begin{tikzpicture}[
    node distance=1.0cm and 1.4cm,
    line width=0.8pt,
    >=latex,
    font=\footnotesize,
    ivblue/.style={
        rectangle,
        rounded corners=2pt,
        draw=black!65,
        fill=blue!8,
        text width=1.05cm,
        align=center,
        minimum height=0.6cm
    },
    ivyellow/.style={
        rectangle,
        rounded corners=2pt,
        draw=black!65,
        fill=yellow!18,
        text width=1.05cm,
        align=center,
        minimum height=0.6cm
    },
    mainnode/.style={
        rectangle,
        rounded corners=2pt,
        draw=black!65,
        fill=white,
        text width=1.05cm,
        align=center,
        minimum height=0.6cm
    },
    uinode/.style={
        rectangle,
        rounded corners=2pt,
        draw=black!65,
        fill=gray!10,
        text width=3.4cm,
        align=center,
        minimum height=1cm
    },
    solidarrow/.style={
        draw=black,
        ->
    }
]

% ---------------- Nodes ----------------

% Covariates
\node[mainnode] (W) at (-3.0,0)
{$\mathbf{W}$};

% IVs
\node[ivyellow] (Dist) at (-1.25,0.8)
{$D_{ist}$};

\node[ivblue] (Alti) at (-1.25,1.7)
{$A_{lti}$};

% Treatment and outcome
\node[mainnode] (Violence) at (0.8,0)
{$V_{io}$};

\node[mainnode] (Patience) at (3.0,0)
{$P_{at}$};

% Unobserved confounder
\node[uinode] (U) at (1.9,1.55)
{\emph{Social and Political}\\
\emph{confounder}};

% ---------------- Edges ----------------

% Covariates -> IVs
\draw[solidarrow] (W) -- (Dist);
\draw[solidarrow] (W) -- (Alti);

% IVs -> treatment
\draw[solidarrow] (Dist) -- (Violence);
\draw[solidarrow] (Alti) -- (Violence);

% Covariates -> treatment
\draw[solidarrow] (W) -- (Violence);

% Treatment -> outcome
\draw[solidarrow] (Violence) -- (Patience);

% Unobserved confounder
\draw[solidarrow] (U) -- (Violence);
\draw[solidarrow] (U) -- (Patience);

% Covariates -> outcome
\draw[solidarrow, bend right=13]
    (W.south) to (Patience.south);

\end{tikzpicture}

\caption{Graphical illustration of an IV model for estimating the causal effect
of \emph{Violence} on a person's \emph{Patience}
\citep{voors2012violent}.}
\label{Fig:real-data-conflict-Preference}
\end{figure}
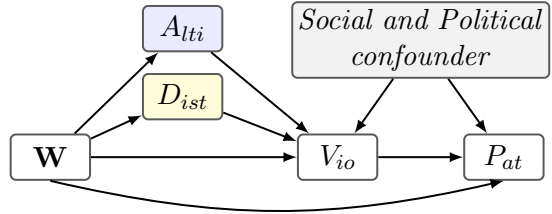

\noindent\textbf{Results.} Since the dataset contains only two candidate IVs, we apply CAT directly to the pair $\{D_{ist},A_{lti}\}$ and obtain a CAT-based distance-correlation score of $\operatorname{dCor}=0.28$. 
We then conduct distance-correlation independence tests for the two directed CAT relations, namely between $\mathcal{V}_{D_{ist}}$ and $A_{lti}$, and between $\mathcal{V}_{A_{lti}}$ and $D_{ist}$, obtaining $p$-values of $0.10$ in both cases. 
Thus, we do not reject the corresponding CAT independence relations, providing no evidence against $\{D_{ist},A_{lti}\}$ as a valid IV set for $V_{io}\to P_{at}$. 
This result is consistent with the findings of~\cite{voors2012violent}.

\section{Conclusion}\label{Sec-conclusion}

In this paper, we studied the problem of testing the validity of IV sets from observational data under causal additive models with non-constant effects (CAM-NCE). 
Under the completeness condition (Assumption~\ref{Ass-completeness}), we introduced a testable necessary condition, termed the Cross Auxiliary-based Independence Test (CAT) condition, for assessing IV set validity. 
Furthermore, under the cross distributional non-degeneracy condition (Assumption~\ref{Ass-higher-order-condition}), we established a necessary and sufficient characterization of valid IV sets within the CAM-NCE framework. 
We also extended the CAT condition to settings with covariates and developed a practical finite-sample algorithm for selecting candidate IV sets that are most consistent with the CAT condition. 
Experimental results on both synthetic and real-world datasets demonstrate the effectiveness and practical utility of the proposed method.
One promising direction for future work is to extend the proposed framework to more general causal models, such as models with multiple treatment variables.

\bibliographystyle{plainnat}
\bibliography{reference.bib}

\clearpage
% \onecolumn
\appendix

\crefalias{section}{appendix}
\crefalias{subsection}{appendix}
\crefalias{subsubsection}{appendix}

\section*{Appendix Contents}
\addcontentsline{toc}{section}{Appendix Contents}
% Start a new ToC that only includes what follows
\startcontents[appendix]
\printcontents[appendix]{}{1}{}

\clearpage
\section{Theoretical Foundations}
Before presenting the proofs, we introduce several technical facts that will be used repeatedly. 
We begin with a standard property of independent random variables from~\cite{meester2008natural}, which is used in the proof of Theorem~\ref{Theorem-Necessary-Condition-IV-set} and its covariate-adjusted extension.

\begin{Theorem}[Theorem 2.2.5 in~\cite{meester2008natural}]\label{The_function_indep} 
Let $X_1, X_2, \ldots, X_n$ be independent random variables, and for $i = 1, \ldots, n$, $g_i$ be a function $g_i: \mathbb{R} \to \mathbb{R}$. Then the random variables $g_1(X_1), g_2(X_2), \ldots, g_n(X_n)$ are also independent. 
\end{Theorem}

We also use the following direct extension, which states that measurable functions of disjoint subsets of mutually independent random variables remain independent.

\begin{Lemma}\label{lemma_multi_function_indep}
Let \( X_1, \dots, X_n \) be independent random variables. Suppose that \( g: \mathbb{R}^k \to \mathbb{R} \) is a measurable function of \( (X_1, \dots, X_k) \), 
\( h: \mathbb{R}^{n-k} \to \mathbb{R} \) is a measurable function of \( (X_{k+1}, \dots, X_n) \), then \( g(X_1, \dots, X_k) \) and \( h(X_{k+1}, \dots, X_n) \) are independent random variables. 
\end{Lemma}
% \begin{proof}
% This follows by applying the same argument using independence and factoring of the joint distribution. 
% \end{proof}

\begin{proof}
Let 
\[
\mathbf{X}_1=(X_1,\dots,X_k), 
\qquad 
\mathbf{X}_2=(X_{k+1},\dots,X_n).
\]
Since \(X_1,\dots,X_n\) are mutually independent, the random vectors
\(\mathbf{X}_1\) and \(\mathbf{X}_2\) are independent. Indeed, for any Borel sets
\(A\subseteq\mathbb{R}^k\) and \(B\subseteq\mathbb{R}^{n-k}\), the independence of
\(X_1,\dots,X_n\) implies
\[
\mathbb{P}(\mathbf{X}_1\in A,\mathbf{X}_2\in B)
=
\mathbb{P}(\mathbf{X}_1\in A)\mathbb{P}(\mathbf{X}_2\in B).
\]
Now, for any Borel sets \(C,D\subseteq\mathbb{R}\), we have
\[
\begin{aligned}
&\mathbb{P}\big(g(\mathbf{X}_1)\in C,\ h(\mathbf{X}_2)\in D\big) \\
&\quad =
\mathbb{P}\big(\mathbf{X}_1\in g^{-1}(C),\ \mathbf{X}_2\in h^{-1}(D)\big).
\end{aligned}
\]
Since \(g\) and \(h\) are measurable, \(g^{-1}(C)\) and \(h^{-1}(D)\) are Borel sets. 
Using the independence of \(\mathbf{X}_1\) and \(\mathbf{X}_2\), we obtain
\[
\begin{aligned}
&\mathbb{P}\big(\mathbf{X}_1\in g^{-1}(C),\ \mathbf{X}_2\in h^{-1}(D)\big) \\
&\quad =
\mathbb{P}\big(\mathbf{X}_1\in g^{-1}(C)\big)
\mathbb{P}\big(\mathbf{X}_2\in h^{-1}(D)\big) \\
&\quad =
\mathbb{P}\big(g(\mathbf{X}_1)\in C\big)
\mathbb{P}\big(h(\mathbf{X}_2)\in D\big).
\end{aligned}
\]
Therefore, $g(X_1,\dots,X_k)$ and $h(X_{k+1},\dots,X_n)$ are independent. 
\end{proof}

Next, we recall a separability result for twice continuously differentiable functions and apply it to log-densities to connect independence with cross second-order partial derivatives. 
This result will be used in the proofs of Proposition~\ref{Pro-Test-Higher-CAM-NCE} and Theorem~\ref{Theorem-necessary-sufficient-condition-CAM-NCE}. 
\begin{Theorem}[\cite{lin1997factorizing}]\label{Theorem-lin}
    The Hessian $H_f$ of function $f$ is block diagonal everywhere, $\partial_i \partial_j f \big|_{\vec{s_0}} = 0$ for all points $\vec{s_0}$ and all $i \le  k$, $j > k$, if and only if f is separable into a sum $f(s_1,...,s_n) = g(s_1,...,s_k) + h(s_{k+1},...,s_n)$ for some functions $g$ and $h$. 
\end{Theorem}

Theorem~\ref{Theorem-lin} implies that if the log-density of two groups of variables is additively separable, then the corresponding cross second-order partial derivatives vanish. 
% Conversely, a nonzero cross second-order partial derivative rules out such additive separability and hence rules out independence under the differentiability condition.

% The above proposition states that function $f$ is separable if and only if its mixed second-order partial derivative is zero. 

% \subsection{Proofs of Theorems, Propositions, and Corollaries}
% \label{app:proofs-main-text}

\section{Proofs}
\subsection{Proof of Theorem \ref{Theorem-Necessary-Condition-IV-set}}\label{proof-theroem-necessary}

\begin{proof}
To prove Theorem \ref{Theorem-Necessary-Condition-IV-set}, we need to show that if the candidate IV set $\mathbf{Z}'$ is a valid IV set relative to $X \to Y$ under CAM-NCE, then for any pair $\{Z_i, Z_j\} \subseteq \mathbf{Z}'$, $\{X, Y||\{Z_i, Z_j\}\}$ will satisfy the CAT condition. 

Under CAM-NCE, the data-generating process is
\begin{equation}\label{Eq_valid_iv_set}
    \begin{aligned}
        X &= g(\mathbf{Z}) + \varphi_{X}(\mathbf{U}) + \varepsilon_{X},\\
        Y &= f(X) + g_Y(\widetilde{\mathbf{Z}}) + \varphi_{Y}(\mathbf{U}) + \varepsilon_{Y},  
    \end{aligned}
\end{equation}
where $\mathbf{U} = \boldsymbol{\varepsilon_U}$.
Since $\mathbf{Z}'$ is a valid IV set, the valid IV set $\mathbf{Z}'$ is not a subset of $\widetilde{\mathbf{Z}}$, and each variable in $\mathbf{Z}'$ is generated solely from its corresponding noise term, i.e., $Z_k = \varepsilon_{Z_k}$, for any $k \in \{1, \dots, |\mathbf{Z}'|\}$. 
% $Z_i = \varepsilon_{Z_i}$, $Z_j = \varepsilon_{Z_j}$. 

Consider any pair of distinct IVs $\{Z_i,Z_j\}\subseteq\mathbf{Z}'$.  According to the definition of auxiliary variable \wrt $X \to Y$ relative to $Z_i$, 
$\mathcal{V}_{X \to Y||Z_i} \coloneqq Y - h_i(X)$, 
where $h_i(\cdot)$ satisfies $\mathbb{E}[\mathcal{V}_{X \to Y||Z_i} | Z_i] = 0 $ and $h_i(\cdot) \neq 0$. 
Because $Z_i$ is a valid IV and Assumption~\ref{Ass-completeness} holds, the standard identification argument for nonparametric IV models implies that the solution is unique and coincides with the structural response function $f(\cdot)$~\citep{newey2003instrumental,bennett2019deep,singh2019kernel}; that is, $h_i(\cdot) = f(\cdot)$.
Similarly, since $Z_j$ is also a valid IV, we have $h_j(\cdot) = f(\cdot)$.

Therefore, we can further express the auxiliary variable as:
\begin{equation}\label{Eq: Theorem-Necessary-Condition-Valid-IV-A}
\begin{aligned}
    \left\{\begin{matrix}
    \mathcal{V}_{X \to Y||Z_i} = Y -h_i(X) = g_Y(\widetilde{\mathbf{Z}}) + \varphi_Y(\mathbf{U}) + \varepsilon_{Y},\\
    \mathcal{V}_{X \to Y||Z_j} = Y - h_j(X) = g_Y(\widetilde{\mathbf{Z}}) + \varphi_Y(\mathbf{U}) + \varepsilon_{Y}.  
    \end{matrix}\right.
\end{aligned}
\end{equation} 
By Theorem \ref{The_function_indep} and its extension Lemma \ref{lemma_multi_function_indep}, if random variables are mutually independent, then any measurable functions applied to disjoint subsets of them yield independent random variables (see Theorem \ref{The_function_indep}, and Lemma~\ref{lemma_multi_function_indep} for further details). Based on this result, we next show that the auxiliary variable $\mathcal{V}_{X \to Y||Z_j}$ and $Z_i$ are statistically independent. Specifically, since all noise terms of variables are mutually independent and $Z_j \in \mathbf{Z}^{\prime} \nsubseteq \widetilde{\mathbf{Z}}$, we can obtain that $\varepsilon_{Z_j}$ is independent of $g_Y(\widetilde{\mathbf{Z}}) + \varphi_Y(\boldsymbol{\varepsilon_U}) + \varepsilon_Y$. Furthermore, combining Equations \eqref{Eq_valid_iv_set} and \eqref{Eq: Theorem-Necessary-Condition-Valid-IV-A}, we conclude that $Z_j$ is independent of $\mathcal{V}_{X \to Y||{Z_i}}$, \ie, $Z_j \CI \mathcal{V}_{X \to Y||{Z_i}}$. 

Likewise, for pairwise ($\mathcal{V}_{X \to Y||Z_j}, Z_i$), we can derive that $\mathcal{V}_{X \to Y||Z_j}$ is independent of $Z_i$, \ie, $Z_i \CI \mathcal{V}_{X \to Y||{Z_j}}$. To sum up, $\{X, Y||\{Z_i, Z_j\}\}$ always satisfies the CAT condition. 

Similarly, the same argument applies to any variable pair in $ \{Z_i,Z_j\} \subseteq \mathbf{Z}'$, implying that every such pair satisfies the CAT condition. Therefore, $\{X, Y||\mathbf{Z}'\}$ satisfies the CAT condition. 
\end{proof}

\subsection{Proof of Proposition \ref{Pro-violate-assumption2}}\label{proof-Coro-same-constant}
\begin{proof}
We prove the proposition for the candidate IV set ${\mathbf{Z}}'\subseteq \mathbf{Z}$. Since the candidate IVs in ${\mathbf{Z}}'$ violate only the exclusion restriction, they are relevant and exogenous, but directly affect the outcome. The data-generating process can be written as 
\begin{equation}\label{Eq:invalid-IV-set-general}
\begin{aligned}
    X &= g_{X}({\mathbf{Z}}') + \phi_X(\mathbf{Z}\setminus{\mathbf{Z}}')
        + \varphi_X(\mathbf U) + \varepsilon_X,\\
    Y &= f(X) + g_{Y}({\mathbf{Z}}') + \phi_Y(\widetilde{\mathbf{Z}}\setminus {\mathbf{Z}}')
        + \varphi_Y(\mathbf U) + \varepsilon_Y, 
\end{aligned}
\end{equation}
where $g_{Y}({\mathbf{Z}}')=a\cdot g_{X}({\mathbf{Z}}')+b$, and $a\neq 0$. 
Substituting this relation into the structural equation of $Y$ gives 
\[
\begin{aligned}
Y&= f(X) + a\cdot g_{X}({\mathbf{Z}}')
   + b + \phi_Y(\mathbf{Z} \setminus {\mathbf{Z}}') + \varphi_Y(\mathbf U) + \varepsilon_Y \\
&= f(X) + a\cdot \{X-\phi_X(\mathbf{Z} \setminus {\mathbf{Z}}')-\varphi_X(\mathbf U)-\varepsilon_X\}\\
   & \quad+ b + \phi_Y(\mathbf{Z} \setminus {\mathbf{Z}}') + \varphi_Y(\mathbf U) + \varepsilon_Y \\
&= f'(X) - a\cdot \{\phi_X(\mathbf{Z}\setminus {\mathbf{Z}}')+\varphi_X(\mathbf U)+\varepsilon_X\} \\
   & \quad + b + \phi_Y(\mathbf{Z} \setminus {\mathbf{Z}}')
   + \varphi_Y(\mathbf U) + \varepsilon_Y,
\end{aligned}
\]
where $f'(X)=f(X)+a\cdot X$. 

Now construct an alternative structural representation with $Z_i'=Z_i$ for $i \in [|\mathbf{Z}|]$, \ie, $\mathbf{Z}_{alter}=\mathbf{Z}$, $X'=X$, and 
\begin{equation}
    \begin{aligned}\nonumber
Y'=&f'(X')- a\cdot \{\phi_X(\{\mathbf{Z}\setminus {\mathbf{Z}}'\}_{alter})+\varphi_X(\mathbf U)+\varepsilon_X\} \\
   & + b + \phi_Y(\{\mathbf{Z} \setminus {\mathbf{Z}}'\}_{alter})
   + \varphi_Y(\mathbf U) + \varepsilon_Y.
    \end{aligned}
\end{equation}
Then $Y'=Y$, and hence 
\[
    (X,Y,Z_1,\ldots,Z_{|\mathbf{Z}|})
    \overset{d}{=}
    (X',Y',Z_1',\ldots,Z_{|\mathbf{Z}|}') .
\]
Furthermore, we have 
\[
    (X,Y,\mathbf{Z}')
    \overset{d}{=}
    (X',Y',\mathbf{Z}_{alter}').
\]
In this alternative representation, the variables $\mathbf{Z}_{alter}'$ affect $Y'$ only through $X'$ and have no direct effects on $Y'$. Therefore, $\mathbf{Z}_{alter}'$ forms a valid IV set with respect to $X'\to Y'$. 

By Theorem~\ref{Theorem-Necessary-Condition-IV-set}, the CAT condition holds for the alternative representation. For any pair $\{Z_i',Z_j'\} \subseteq \mathbf{Z}_{alter}'$ with $i\neq j$, the corresponding auxiliary variable is 
\begin{equation}\nonumber
    \begin{aligned}
        \mathcal{V}_{X'\to Y'\|Z_i'} &= Y'-f'(X') \\
        &= \phi_Y(\{\mathbf{Z}\setminus {\mathbf{Z}}'\}_{alter}) + \varphi_Y(\mathbf U) + \varepsilon_Y  + b\\
        &\quad - a\{\phi_X(\{\mathbf{Z}\setminus {\mathbf{Z}}'\}_{alter})+\varphi_X(\mathbf U)+\varepsilon_X\}.
    \end{aligned}
\end{equation}

Since each $Z_j'$ is independent of $\mathbf U$, $\varepsilon_X$, $\varepsilon_Y$, and other variables $\{\mathbf{Z}\setminus\mathbf{Z}'\}_{alter}$, it follows from Lemma \ref{lemma_multi_function_indep} that 
\[ \mathcal V_{X'\to Y'\|Z_i'} \CI Z_j', \qquad \forall i\neq j. \] 
Consequently, $p(\mathcal V_{X'\to Y'\|Z_i'},Z_j') = p(\mathcal V_{X'\to Y'\|Z_i'}) \,p(Z_j')$, which implies 
\[
\log p(\mathcal V_{X'\to Y'\|Z_i'}, Z_j') = \log p(\mathcal V_{X'\to Y'\|Z_i'}) + \log p(Z_j').
\]
Therefore, 
\[ \frac{\partial^2 \log p(\mathcal V_{X'\to Y'\|Z_i'},Z_j')} {\partial \mathcal V_{X'\to Y'\|Z_i'}\partial Z_j'} =0. \] 
Because the original and alternative representations induce the same observational distribution, we also have \[ \frac{\partial^2 \log p(\mathcal V_{X\to Y\|Z_i},Z_j)} {\partial \mathcal V_{X\to Y\|Z_i}\partial Z_j} =0, \qquad \forall i\neq j. \]
The same argument applies after exchanging $Z_i$ and $Z_j$. Consequently, all the cross second-order partial derivatives are zero. As a result, $\{X,Y||\mathbf{Z}'\}$ violates 
Assumption~\ref{Ass-higher-order-condition}. 
\end{proof}

\subsection{Proof of Proposition \ref{Pro-Test-Higher-CAM-NCE}}\label{prooof-Pro-Test-Higher-CAM-NCE}

\begin{proof}
Suppose that the candidate IV set $\mathbf{Z}'$ is invalid. 
By Assumption~\ref{Ass-higher-order-condition}, there exists a pair of distinct candidate IVs $\{Z_i,Z_j\}\subseteq\mathbf{Z}'$ such that the joint densities 
$p(\mathcal{V}_{X\to Y\|Z_i},Z_j)$ and 
$p(\mathcal{V}_{X\to Y\|Z_j},Z_i)$ are twice continuously differentiable, and at least one of the following cross second-order partial derivatives is nonzero on a set with nonzero Lebesgue measure:
\[
\frac{\partial^2 \log p(\mathcal{V}_{X\to Y\|Z_i},Z_j)}
{\partial \mathcal{V}_{X\to Y\|Z_i}\,\partial Z_j}
\quad \text{or} \quad
\frac{\partial^2 \log p(\mathcal{V}_{X\to Y\|Z_j},Z_i)}
{\partial \mathcal{V}_{X\to Y\|Z_j}\,\partial Z_i}.
\]

We prove the result by contradiction. 
Assume that $\{X,Y\|\mathbf{Z}'\}$ satisfies the CAT condition. 
Then, for every pair of distinct IVs $\{Z_i,Z_j\}\subseteq\mathbf{Z}'$, we have
\[
    \mathcal{V}_{X\to Y\|Z_i}\CI Z_j,
    \qquad
    \mathcal{V}_{X\to Y\|Z_j}\CI Z_i.
\]
In particular, for the pair specified by Assumption~\ref{Ass-higher-order-condition}, the first independence relation implies
\[
p(\mathcal{V}_{X\to Y\|Z_i},Z_j)
=
p(\mathcal{V}_{X\to Y\|Z_i})p(Z_j).
\]
Taking logarithms yields the additive decomposition
\[
\log p(\mathcal{V}_{X\to Y\|Z_i},Z_j)
=
\log p(\mathcal{V}_{X\to Y\|Z_i})+\log p(Z_j).
\]
Therefore, the corresponding cross second-order partial derivative must vanish:
\[
\frac{\partial^2 \log p(\mathcal{V}_{X\to Y\|Z_i},Z_j)}
{\partial \mathcal{V}_{X\to Y\|Z_i}\,\partial Z_j}
=
0.
\]
Similarly, the second independence relation 
$\mathcal{V}_{X\to Y\|Z_j}\CI Z_i$ implies
\[
\frac{\partial^2 \log p(\mathcal{V}_{X\to Y\|Z_j},Z_i)}
{\partial \mathcal{V}_{X\to Y\|Z_j}\,\partial Z_i}
=
0.
\]
Thus, both cross second-order partial derivatives vanish, which contradicts Assumption~\ref{Ass-higher-order-condition}. 
Hence, $\{X,Y\|\mathbf{Z}'\}$ cannot satisfy the CAT condition. 
Therefore, if $\mathbf{Z}'$ is invalid, then $\{X,Y\|\mathbf{Z}'\}$ violates the CAT condition.
\end{proof}

\subsection{Proof of Theorem \ref{Theorem-necessary-sufficient-condition-CAM-NCE}}\label{proof-Theorem-necessary-sufficient-condition-CAM-NCE}
\begin{proof}
We prove the necessary and sufficient characterization of valid IV sets under CAM-NCE by establishing the following two implications.

(i): Assume the candidate IV set $\mathbf{Z}'$ is a valid IV set relative to $X \to Y$. By Theorem \ref{Theorem-Necessary-Condition-IV-set}, under Assumption \ref{Ass-completeness}, it directly follows that if the candidate IV set $\mathbf{Z}'$ is a valid IV set relative to $X \to Y$, then $\{X, Y||\mathbf{Z}'\}$ always satisfies the CAT condition.

(ii): Assume the candidate IV set $\mathbf{Z}'$ is an invalid IV set relative to $X \to Y$. By Proposition \ref{Pro-Test-Higher-CAM-NCE}, under Assumptions \ref{Ass-completeness} and \ref{Ass-higher-order-condition}, if the candidate IV set $\mathbf{Z}'$ is invalid, then $\{X, Y||\mathbf{Z}'\}$ consequently violates the CAT condition. 

Combining (i) and (ii), we conclude that $\mathbf{Z}'$ is a valid IV set w.r.t. $X \to Y$ if and only if $\{X,Y\|\mathbf{Z}'\}$ satisfies the CAT condition.
\end{proof}

\subsection{Proof of Corollary \ref{Corollary-Necessary-Condition-IV-set-covariates}}\label{proof-Corollary-Necessary-Condition-IV-set-covariates}

\begin{proof}
Suppose that $\mathbf{Z}'$ is a valid IV set w.r.t. $X \to Y$ given $\mathbf{W}$. 
Under CAM-NCE with covariates, the data-generating mechanism can be written as
\begin{equation}\label{Eq:CAM-NCE-covariates}
\begin{aligned}
    \mathbf{U} &= \boldsymbol{\varepsilon_U}, 
    \qquad 
    \mathbf{W} = t_W(\mathbf{PA}_{\mathbf{W}}) + \boldsymbol{\varepsilon_W},  \\ 
    \mathbf{Z}' &= \psi_{\mathbf{Z}'}(\mathbf{Z}')+ t_{\mathbf{Z}'}(\mathbf{W})  + \varepsilon_{\mathbf{Z}'}, \\
    X &= g_X(\mathbf{Z}) + t_X(\mathbf{W}) + \varphi_X(\mathbf{U}) + \varepsilon_X, \\ 
    Y &= f(X) + t_Y(\mathbf{W}) + g_Y(\widetilde{\mathbf{Z}}) + \varphi_Y(\mathbf{U}) + \varepsilon_Y,
\end{aligned}
\end{equation}
where $\mathbf{PA}_{\mathbf{W}}$ denotes the set of parent variables of $\mathbf{W}$, and $\psi_{\mathbf{Z}'}$ denotes the causal relationship among variables in $\mathbf{Z}'$. 
Since $\mathbf{Z}'$ is valid given $\mathbf{W}$, no variable in $\mathbf{Z}'$ directly affects $Y$; equivalently, 
\[
    \mathbf{Z}' \cap \widetilde{\mathbf{Z}} = \emptyset .
\]

Let $\mathcal{X}$, $\mathcal{Y}$, and $\boldsymbol{\mathcal{Z}}'$ denote the residuals obtained by regressing $X$, $Y$, and $\mathbf{Z}'$ on $\mathbf{W}$, respectively. 
Under the additive covariate structure in Equation \eqref{Eq:CAM-NCE-covariates} and the regression adjustment in Definition~\ref{Definition-CAT-condition-set-covariates}, the residualized variables remove the effects of $\mathbf{W}$ and follow the same structural form as the covariate-free CAM-NCE model. 
Moreover, because $\mathbf{Z}'$ is a valid IV set w.r.t. $X\to Y$ given $\mathbf{W}$, the residualized candidate IV set $\boldsymbol{\mathcal{Z}}'$ satisfies the corresponding relevance, exclusion restriction, and exogeneity conditions in the residualized system. 

Therefore, applying Theorem~\ref{Theorem-Necessary-Condition-IV-set} to the residualized variables yields that every pair of distinct residualized IVs in $\boldsymbol{\mathcal{Z}}'$ satisfies the CAT condition. 
Equivalently, $\{X,Y\|(\mathbf{Z}',\mathbf{W})\}$ satisfies the CAT condition.
\end{proof}

\subsection{Proof of Corollary \ref{Corollary-necessary-sufficient-condition-CAM-NCE-covariates}}\label{proof-Corollary-necessary-sufficient-condition-CAM-NCE-covariates}

\begin{proof}
Let $\mathcal{X}$, $\mathcal{Y}$, and $\boldsymbol{\mathcal{Z}}'$ denote the residuals obtained by regressing $X$, $Y$, and $\mathbf{Z}'$ on $\mathbf{W}$, respectively. 
As shown in the proof of Corollary~\ref{Corollary-Necessary-Condition-IV-set-covariates}, under the additive covariate structure and the regression adjustment in Definition~\ref{Definition-CAT-condition-set-covariates}, the residualized variables $\mathcal{X}$, $\mathcal{Y}$, and $\boldsymbol{\mathcal{Z}}'$ follow the same structural form as the covariate-free CAM-NCE model.

Furthermore, Assumption~\ref{Ass-completeness} and Assumption~\ref{Ass-higher-order-condition} are assumed to hold for these covariate-adjusted residual variables. 
Therefore, the residualized system satisfies the conditions required by Theorem~\ref{Theorem-necessary-sufficient-condition-CAM-NCE}. 
Applying Theorem~\ref{Theorem-necessary-sufficient-condition-CAM-NCE} to the residualized variables yields that the residualized candidate IV set $\boldsymbol{\mathcal{Z}}'$ is valid for the causal relation $\mathcal{X}\to\mathcal{Y}$ if and only if the corresponding CAT condition holds.

By Definition~\ref{Definition-CAT-condition-set-covariates}, this is equivalent to saying that $\mathbf{Z}'$ is a valid IV set w.r.t. $X\to Y$ given $\mathbf{W}$ if and only if $\{X,Y\|(\mathbf{Z}',\mathbf{W})\}$ satisfies the CAT condition.
\end{proof}

\subsection{Proof of Theorem \ref{The-algorithm-correctness}}\label{proof-The-algorithm-correctness}

\begin{proof}
Let \(\mathbf Z_V\subseteq \mathbf Z\) denote the set of all valid IVs among the
candidate IVs, and let \(K_V=|\mathbf Z_V|\). By assumption, \(K_V\ge K\).
We assume \(K\ge 2\), so that the pairwise CAT criterion is informative.

Algorithm~\ref{algorithm-CAT} first removes the effect of covariates
\(\mathbf W\), if present, and then constructs, for each candidate IV
\(\mathcal Z_i\), the auxiliary variable $\widehat{\mathcal V}_{\mathcal X\to \mathcal Y\|\mathcal Z_i}$. By the assumed consistency of the estimators used for covariate adjustment,
auxiliary-variable construction, and distance correlation, for every unordered
pair \(\{\mathcal Z_i,\mathcal Z_j\}\subseteq \boldsymbol{\mathcal{Z}}\),
\[
\widehat{\mathbf{M}}_{ij}
=
\widehat{\operatorname{dCor}}
\bigl(
\widehat{\mathcal V}_{\mathcal X\to \mathcal Y\|\mathcal Z_i},
\widehat{\mathcal Z}_j
\bigr)
+
\widehat{\operatorname{dCor}}
\bigl(
\widehat{\mathcal V}_{\mathcal X\to \mathcal Y\|\mathcal Z_j},
\widehat{\mathcal Z}_i
\bigr)
\]
converges in probability to
\[
\mathbf{M}_{ij}
=
\operatorname{dCor}
\bigl(
\mathcal V_{\mathcal X\to \mathcal Y\|\mathcal Z_i},
\mathcal Z_j
\bigr)
+
\operatorname{dCor}
\bigl(
\mathcal V_{\mathcal X\to \mathcal Y\|\mathcal Z_j},
\mathcal Z_i
\bigr).
\]

For any candidate subset \(\mathcal S_c\subseteq \boldsymbol{\mathcal{Z}}\) with
\(|\mathcal S_c|=K\), define the population objective
\[
T_{\mathcal S_c}
=
\sum_{\substack{\{\mathcal Z_i,\mathcal Z_j\}\subseteq \mathcal S_c\\ i<j}}
\mathbf{M}_{ij},
\]
and let \(\widehat T_{\mathcal S_c}\) be the corresponding empirical objective
computed by Algorithm~\ref{algorithm-CAT}. Since the number of candidate subsets
of size \(K\) is finite and each \(\widehat{\mathbf{M}}_{ij}\) is consistent, we have the
uniform convergence
\[
\max_{\mathcal S_c:\,|\mathcal S_c|=K}
\bigl|
\widehat T_{\mathcal S_c}-T_{\mathcal S_c}
\bigr|
\overset{p}{\longrightarrow}0.
\]

Now consider any subset \(\mathcal S_c\subseteq \boldsymbol{\mathcal{Z}}_V\) with
\(|\mathcal S_c|=K\). Every element of \(\mathcal S_c\) is a valid IV. Hence,
by Theorem~\ref{Theorem-necessary-sufficient-condition-CAM-NCE}, the CAT
condition holds for every pair
\(\{\mathcal Z_i,\mathcal Z_j\}\subseteq \mathcal S_c\), namely
\[
\mathcal V_{\mathcal X\to \mathcal Y\|\mathcal Z_i}\CI \mathcal Z_j,
\qquad
\mathcal V_{\mathcal X\to \mathcal Y\|\mathcal Z_j}\CI \mathcal Z_i.
\]
Since distance correlation is zero if and only if independence holds, it follows
that $\mathbf{M}_{ij}=0$ for every pair in \(\mathcal S_c\). Therefore, $T_{\mathcal S_c}=0$. 

Conversely, let \(\mathcal S_c\) be any subset of size \(K\) that contains at
least one invalid IV. Then \(\mathcal S_c\) is not a valid IV set. By the
necessity and sufficiency result in
Theorem~\ref{Theorem-necessary-sufficient-condition-CAM-NCE}, under
Assumptions~\ref{Ass-completeness}--\ref{Ass-higher-order-condition}, the CAT
condition cannot hold for all pairs in \(\mathcal S_c\). Hence there exists at
least one pair \(\{\mathcal Z_i,\mathcal Z_j\}\subseteq \mathcal S_c\) such that
\[
\mathcal V_{\mathcal X\to \mathcal Y\|\mathcal Z_i}\nCI \mathcal Z_j
\quad
\text{or}
\quad
\mathcal V_{\mathcal X\to \mathcal Y\|\mathcal Z_j}\nCI \mathcal Z_i .
\]
For this pair, at least one of the two distance correlations is strictly positive. Thus $\mathbf{M}_{ij}>0$, and consequently $T_{\mathcal S_c}>0$. 

Therefore, every valid subset of size \(K\) attains the population minimum
value \(0\), whereas every subset of size \(K\) containing at least one invalid
IV has a strictly positive population objective value.

Because the number of candidate subsets is finite, the collection of invalid
subsets of size \(K\) is also finite. If this collection is nonempty, define
\[
\Delta
=
\min_{\substack{\mathcal S_c:\,|\mathcal S_c|=K\\
\mathcal S_c\not\subseteq \mathbf Z_V}}
T_{\mathcal S_c}.
\]
From the preceding argument, \(\Delta>0\). On the event
\[
\max_{\mathcal S_c:\,|\mathcal S_c|=K}
\bigl|
\widehat T_{\mathcal S_c}-T_{\mathcal S_c}
\bigr|
<\frac{\Delta}{2},
\]
every valid subset \(\mathcal S_c\subseteq \boldsymbol{\mathcal{Z}}_V\) satisfies $\widehat T_{\mathcal S_c}<\frac{\Delta}{2}$, whereas every invalid subset satisfies $\widehat T_{\mathcal S_c}>\frac{\Delta}{2}$. Hence any empirical minimizer selected in Line~26 of Algorithm~\ref{algorithm-CAT} must be a subset of \(\boldsymbol{\mathcal{Z}}_V\). Since the
above event has probability tending to one (\ie, $\mathbb{P}(\max_{\mathcal S_c:\,|\mathcal S_c|=K}
\bigl|
\widehat T_{\mathcal S_c}-T_{\mathcal S_c}
\bigr|
<\frac{\Delta}{2})\to 1$), the algorithm outputs, with probability tending to one, a valid IV subset \(\widehat{\mathcal S} \subseteq \boldsymbol{Z}_V\) with $|\widehat{\mathcal S}|=K$. 

Finally, if the candidate set \(\mathbf Z\) contains exactly \(K\) valid IVs,
then \(K_V=K\), and the only subset of size \(K\) consisting entirely of valid
IVs is \(\mathbf Z_V\) itself. Therefore,
\[
\widehat{\mathcal S}\overset{p}{\longrightarrow}\mathbf Z_V,
\]
in the sense that the probability that Algorithm~\ref{algorithm-CAT} outputs the full valid IV set tends to one. 

This proves the stated correctness of Algorithm~\ref{algorithm-CAT}.
\end{proof}

\end{document}